\documentclass[envcountsame,oribibl,orivec]{llncs}
\usepackage{amsmath,amssymb}
\usepackage{graphicx,wrapfig,float,url,booktabs}
\usepackage[usenames]{xcolor}
\usepackage[pdfborder=0 0 0]{hyperref}
\usepackage[english]{babel}
\usepackage{tabularx,xspace}
\usepackage{algo}
\usepackage{caption}
\usepackage[caption=false]{subfig}
\usepackage{cite}
\usepackage{appendix}
\usepackage{todonotes}
\graphicspath{{images/}}
\usepackage[margin=2.5cm]{geometry}

\renewcommand{\paragraph}[1]{\smallskip\noindent{\bfseries #1.}}

\title{
Distribution-Sensitive Construction of the Greedy Spanner
}

\author{Sander P. A. Alewijnse \and
  Quirijn W. Bouts\thanks{Q. W. Bouts is supported by the Netherlands Organisation for Scientific Research (NWO) under project no.~639.023.208} \and
        Alex P. ten Brink \and Kevin Buchin
}

\institute{
Eindhoven University of Technology, The Netherlands,\\
\path|q.w.bouts@tue.nl|
}

\expandafter\ifx\csname pdfpagebox\endcsname\relax\else
\fi

\spnewtheorem{observation}[theorem]{Observation}{\bfseries}{\itshape}
\spnewtheorem{fact}[theorem]{Fact}{\bfseries}{\itshape}

{\unskip\nobreak\hskip 1em plus 1fil\nobreak$\Box$
\parfillskip=0pt%
\endtrivlist}
{\unskip\nobreak\hskip 2em plus 1fil\nobreak$\Box$
\parfillskip=0pt%
\endtrivlist}

\let\doendproof\endproof
\renewcommand\endproof{~\hfill$\qed$\doendproof}

\newcommand{\safespace}{}

\newcommand{\Reals}{\mathbb{R}}

\begin{document}
\mainmatter
\maketitle
\safespace
\begin{abstract}
The greedy spanner is the highest quality geometric spanner (in e.g. edge count and weight, both in theory and practice) known to be computable in polynomial time. Unfortunately, all known algorithms for computing it on $n$ points take $\Omega(n^2)$ time, limiting its use on large data sets.

We observe that for many point sets, the greedy spanner has many `short' edges that can be determined locally and usually quickly, and few or no `long' edges that can usually be determined quickly using local information and the well-separated pair decomposition. We give experimental results showing large to massive performance increases over the state-of-the-art on nearly all tests and real-life data sets. On the theoretical side we prove a near-linear expected time bound on uniform point sets and a near-quadratic worst-case bound.

Our bound for point sets drawn uniformly and independently at random in a square follows from a local characterization of $t$-spanners we give on such point sets: we give a geometric property that holds with high probability on such point sets. This property implies that if an edge set on these points has $t$-paths between pairs of points `close' to each other, then it has $t$-paths between all pairs of points.

This characterization gives a $O(n \log^2 n \log^2 \log n)$ expected time bound on our greedy spanner algorithm, making it the first subquadratic time algorithm for this problem on any interesting class of points. We also use this characterization to give a $O((n + |E|) \log^2 n \log \log n)$ expected time algorithm on uniformly distributed points that determines if $E$ is a $t$-spanner, making it the first subquadratic time algorithm for this problem that does not make assumptions on $E$.
\end{abstract}

\vspace{-1\baselineskip}
\section{Introduction}
\safespace

A \emph{Euclidean graph} on a set of $n$ points in the Euclidean plane is a weighted graph with geometric distances as edge weights. If a shortest route in the graph is at most $t$ times longer than the direct geometric distance between its endpoints, we say these endpoints \emph{have a $t$-path}: a Euclidean graph is a $t$-spanner if all pairs of points have $t$-paths. For any $t>1$, we can efficiently find a $t$-spanner with $O\left(\frac{n}{t-1}\right)$ edges in the Euclidean plane \cite{Narasimhan:2007:GSN:1208237}. These `approximations' have few edges compared to the complete graph, while approximately maintaining distances, making them a useful tool in many areas. 

Bounded degree spanners are used in wireless network design~\cite{GaoGHZZ05}, where for example points of high degree tend to have problems with interference. By using such a bounded degree spanner the problem of interference is minimized while the connectivity is maintained. A considerable amount of research has been done on spanners~\cite{DilationandDetours,Narasimhan:2007:GSN:1208237} since they were introduced in network design~\cite{JGT:JGT3190130114} and in geometry~\cite{Chew1989}. Spanners have been used as components in various geometric and distributed algorithms.

Many different construction methods exist for $t$-spanners, where $t$ can be parameterized to an arbitrary value greater than 1, each having different advantages and disadvantages. An in-depth treatise of these spanners can be found in the book~\cite{Narasimhan:2007:GSN:1208237}. We focus on the greedy spanner, which is defined as the graph resulting from repeatedly adding the edge between the closest pair of points which do not have a $t$-path yet. The result is a very sparse graph with assymptotically optimal edge count, degree and weight. On uniform point sets and for $t=2$, one of its closest well-known competitors with respect to these three properties is the $\Theta$-graph. It has about ten times as many edges, twenty times higher total weight and six times higher maximum degree. Figure~\ref{figure:greedy} clearly shows the contrast between these two spanners. Unfortunately, all known algorithms computing the greedy spanner use $\Omega(n^2)$ time\cite{BoseCFMS2010,
AlewijnseBBB13}, making the spanner impractical to compute.

\begin{wrapfigure}[16]{r}{0.55\textwidth} %
\raggedleft
\begin{minipage}{0.53\textwidth}
 \includegraphics[width=0.495\textwidth]{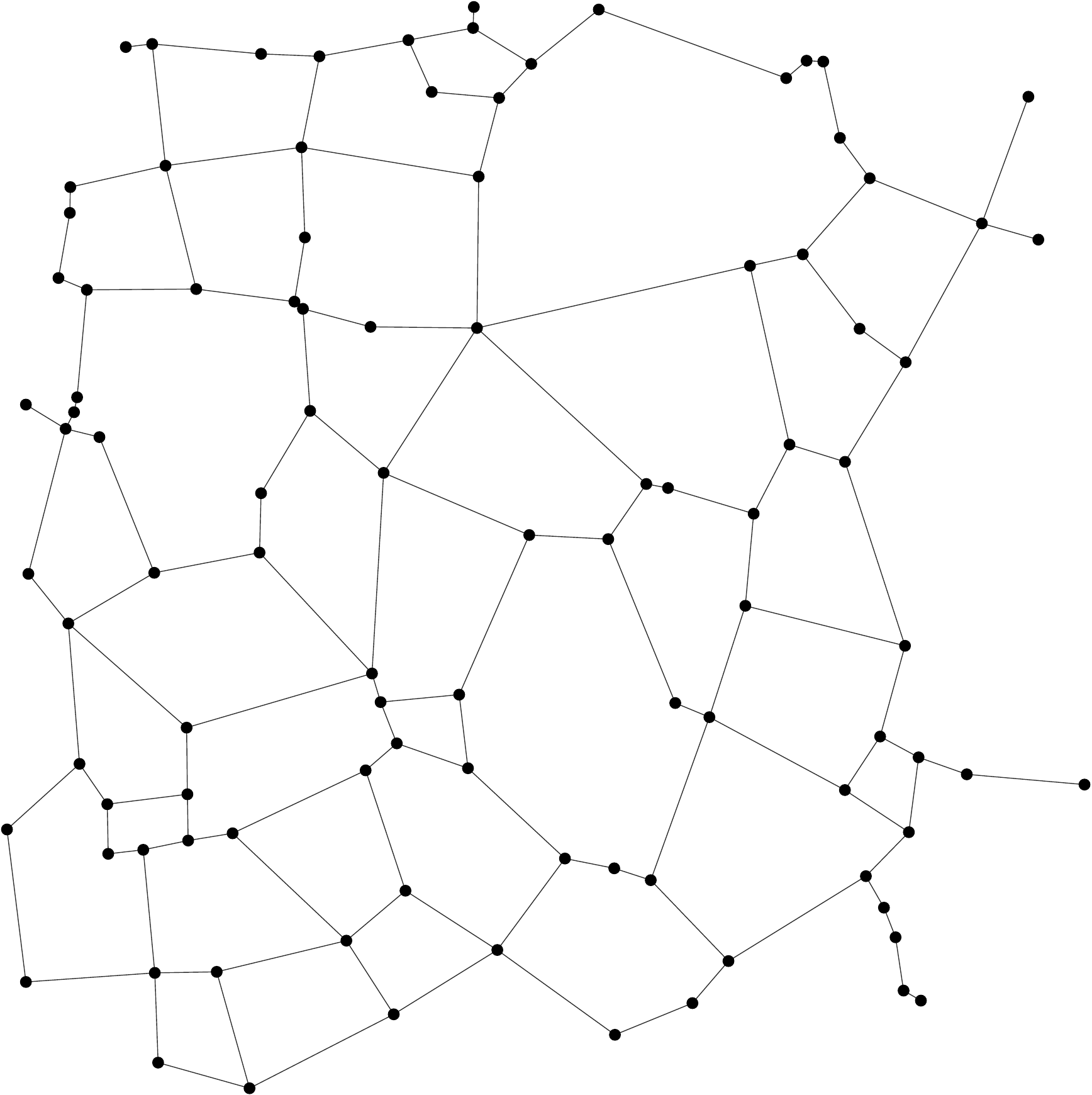}
 \includegraphics[width=0.495\textwidth]{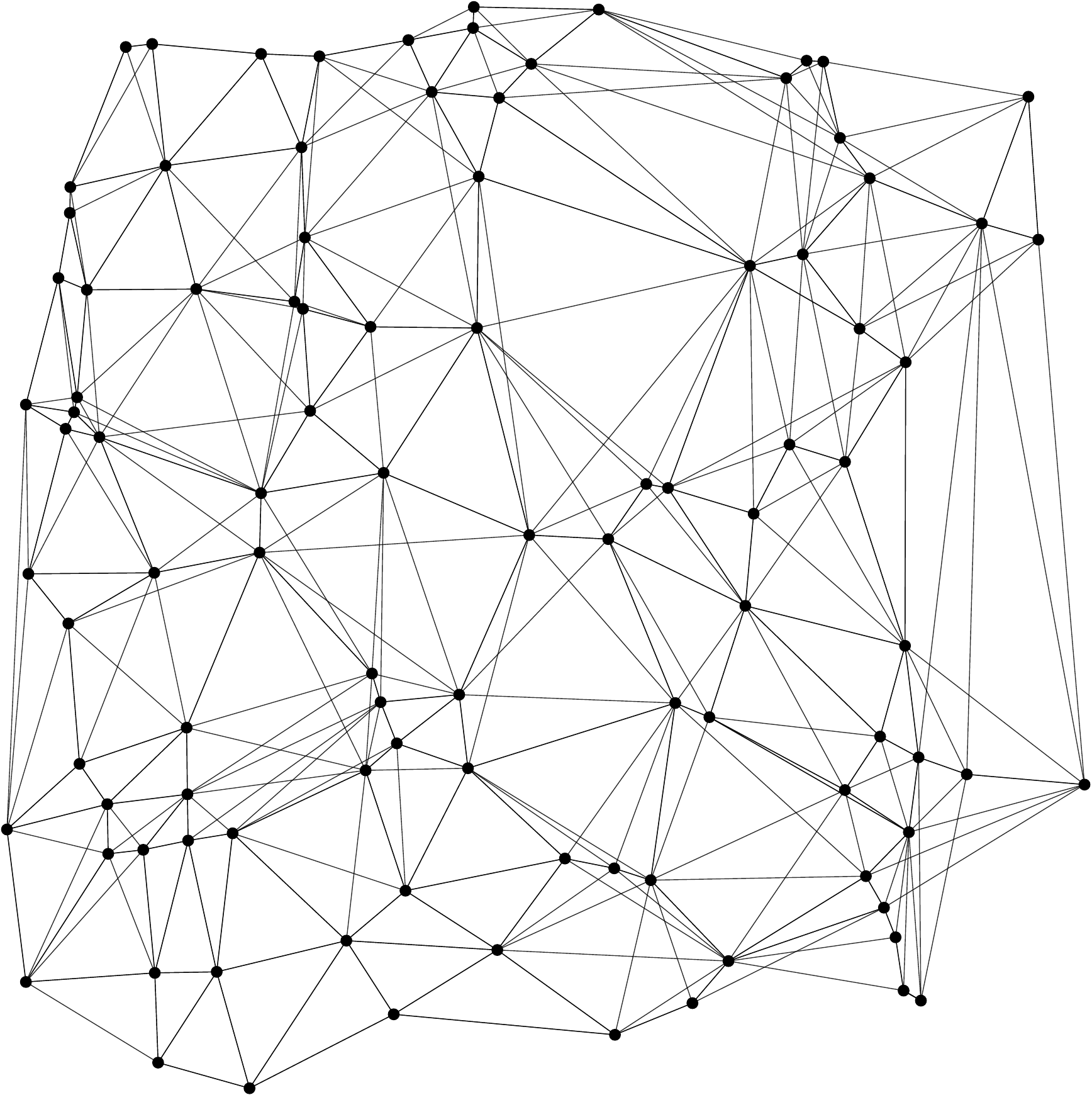}
\end{minipage}
\vspace{-.5\baselineskip}
 \caption{The left rendering shows the greedy spanner on 100 points distributed uniformly in a square with $t=2$. The right rendering shows the $\Theta$-graph on the same points with $k=6$ for which it was recently proven it achieves a dilation of 2.}
 \vspace{-.5em}
 \label{figure:greedy}
\end{wrapfigure}

We observed that on real-world examples, the greedy spanner contains mostly short edges with at most a few longer edges. Whether an edge is placed depends only on the points and edges in an ellipse with its endpoints as foci and with eccentricity $1/t$, which is a small area for short potential edges, hopefully containing few points. We can therefore find these short edges using a bucketing scheme, giving a speedup on such point sets.

For the `long' edges, we consider the `long' well-separated pairs from a Well-separated pair decomposition (WSPD)~\cite{Callahan95dealingwith}. We first compute information from the `short' edges, attempting to find witnesses that show that certain `long' well-separated pairs will not contain greedy spanner edges. This information is represented by \emph{path-hyperbola}. We then perform a standard algorithm~\cite{AlewijnseBBB13} on the (hopefully only few) well-separated pairs for which we cannot find such a witness.

We present experimental results showing that the above algorithm works very well on many data sets, ranging from real-world data sets to sets which are generated according to different distributions. Speedups vary from an (apparently) linear factor to a constant factor. In particular, on a uniformly distributed point set with 300,000 points, our new algorithm needs 19 minutes to compute the greedy spanner for $t=2$, while the only other algorithm that can handle point sets of this size \cite{AlewijnseBBB13} (other algorithms need quadratic space, which is prohibitive) needs 17 hours on the same set.

We show that our algorithm has a near-quadratic worst-case time bound. We give formal evidence for the algorithm's good behavior observed in experiments on realistic point sets (which are often reasonably spread out) by analyzing its performance on point sets distributed uniformly and independently at random in a square (or `uniformly distributed points' for short).

Euclidean graphs are frequently analyzed on uniformly distributed points, both concerning theoretical properties and experimental evaluation of structures and algorithms. One can find examples in computational geometry~\cite{buchin2009constructing,mucke1996fast}, combinatorial optimization~\cite{steele1997probability,yukich1998probability} and the analysis of ad-hoc networks~\cite{santi2005topology,Xue:2004}.

Various spanner constructions have been analyzed on uniformly distributed point sets~\cite{abam2009region,Bose:2006,devroye2009expected,Shpungin:2010,wang2007efficient}. Some of these constructions are a $t$-spanner for fixed $t$, others are parameterizable with arbitrary $t > 1$. Relatively sharp bounds have been obtained on various qualities of these spanners. This gives insight into the behavior of these constructions in situations arguably closer to realistic point sets than worst case situations.

The spanner constructions studied in these analyses have a `local' characterization: for example, Gabriel graphs connect $u, v$ if the circle having $uv$ as its diameter contains no points other than $u$ and $v$. For graphs with such a local characterization there are well-developed techniques to analyze them on uniformly distributed points~\cite{geom-150}. In this paper, however, we look at the `global' property $t$-spannerness and the greedy spanner, a graph for which the existence of an edge may depend on all other points. Previous analysis techniques do not directly apply on such properties. However, one of our main contributions is to show that with high probability, greedy spanners do admit a local characterization on uniform point sets.

We give two more examples of local analysis. For a pair of points the minimum $t$ such that there is a $t$-path between them is called their dilation. In a $t$-spanner for all pairs of points the dilation is bounded by $t$. For graphs on points drawn from a Poisson point process also the average dilation between pairs of points has been studied. 
Many graphs with a local characterization like the Gabriel graph have low average dilation~\cite{aldous2009connected}. The property of having low average dilation can be linked to percolation~\cite{Aldous:Shun:2010:1}.

We consider points distributed uniformly and independently at random in a $\sqrt{n} \times \sqrt{n}$ square. We use 
 this square so that if 
  we have an area $A$, then $O(A)$ points lie in it in expectation. We only consider the case of the Euclidean plane -- our results may generalize to higher dimensions, but we did not explore this. In this introduction, when stating bounds, we assume $t$ is a constant.

We prove that such point sets are, with high probability, configured in such a way that for any edge set $E$, if there are $t$-paths between points at most $O(\log n)$ away from each other, then there are $t$-paths between all points. In particular, we show that we can construct a `witness' of this configuration in $O(n \log^2 n \log \log n)$ expected time if it exists, thus allowing our algorithms to always give the correct answer.

This result easily implies that with high probability the greedy spanner has no long edges (longer than $O(\log n)$) and furthermore that the `proof' phase of our algorithm will find the witnesses for this if it exists. As the grid strategy works well on uniformly distributed point sets, we obtain a $O(n \log^2 n \log^2 \log n)$ expected time bound on our algorithm. To the best of our knowledge, this algorithm is the first subquadratic algorithm to compute the greedy spanner on any interesting class of point sets.

Another application of our result is a method to test whether a Euclidean graph $G=(P, E)$ is a $t$-spanner on uniformly distributed points in $O((n + |E|) \log^2 n \log \log n)$ expected time. 
Various algorithms are known for specific graphs on arbitrary points, but not for arbitrary graphs on specific sets of points.
Hellweg et al.~\cite{hss-tes-11} give a Monte Carlo algorithm for bounded degree graphs that distinguishes between being a $t$-spanner and being far away from a spanner. 
For specific graph classes the minimum $t$ can be computed~\cite{Agarwal:2008:CDS:1349672.1349682,eppstein2007minimum}, and for general graphs this $t$ can be approximated~\cite{Narasimhan:2000:ASF:586846.586967}.

The rest of the paper is organized as follows. In Section~\ref{section:property} we introduce \emph{bridgedness} and give a geometric lemma that will help us obtain our results. In Section~\ref{section:uniformpoints} we show uniform point sets are locally-$O(\log n)$-bridged with high probability. In Section~\ref{section:consequences} we give several fast algorithms that use this result. Finally, in Section~\ref{section:results} we present experimental results for our algorithm that computes the greedy spanner.

\section{Bridging Points} \label{section:property}

In this section we will introduce the concept of $\lambda$-bridgedness for point sets. We will later use this concept in our characterization of $t$-spanners on uniformly distributed point sets. We prove two geometric lemmas that will help us with the result of Section~\ref{section:uniformpoints}.

Let $P$ be a finite set of points in $\Reals^2$, let $n=|P|$, and let $t \in \Reals$ be the intended dilation ($t > 1$). Let $G = (P, E)$ be a graph on $P$ whose edges are weighted with the Euclidean distance between its endpoints. For two points $u, v \in P$, we denote the Euclidean distance between $u$ and $v$ by $|uv|$, and the network distance in $G$ by $\delta_G(u, v)$ (or just 
$\delta(u, v)$ if $G$ is clear from the context). We say a pair of points $(u, v)$ \emph{has a $t$-path} if $\delta(u, v) \leq t \cdot |uv|$. If all pairs of points have a $t$-path, the graph is called a \emph{$t$-spanner}.

Let $a, b, p, q \in P$ be pairwise different points. We say that the pair $(p, q)$ \emph{bridges} the pair $(a, b)$ if $t \cdot |ap| + |pq| + t \cdot |qb| \leq t \cdot |ab|$. Bridging points guarantee a $t$-path for $(a, b)$ if $(p, q)$ is an edge and the pairs $(a, p)$ and $(q, b)$ already have $t$-paths. Note that $|ap|, |qb| < |ab|$ as a consequence.

We say that $(p, q)$ is \emph{mandatory} if the ellipse with foci $p$ and $q$ and eccentricity $1/t$ including its border 
contains no points in $P$ other than $p$ and $q$. Any $t$-path between $p$ and $q$ must fully lie within this ellipse, so a mandatory $(p,q)$ will be in $E$ for any $t$-spanner.

Let $\lambda \in \Reals$. We say that a point $a \in P$ is \emph{$\lambda$-bridged} if for all $b \in P$ with $|ab| > \lambda$, there exist some mandatory pair of points $(p, q)$, $p, q \in P$, bridging $(a, b)$. We say that the point set $P$ is \emph{$\lambda$-bridged} if all points in $P$ are $\lambda$-bridged. We say a point $a \in P$ is \emph{locally-$\lambda$-bridged} if it is $\lambda$-bridged using only mandatory bridging pairs of points at with distance most $\lambda$ from $a$. A point set $P$ is \emph{locally-$\lambda$-bridged} if all points in $P$ are locally-$\lambda$-bridged. Lemma~\ref{lemma:bridgedness} shows the usefulness of this concept.
In Lemma~\ref{lemma:nicebox} we give a sufficient geometric condition for bridging pairs of points.

\begin{lemma}
\label{lemma:bridgedness}
Let $P$ be a set of points that is $\lambda$-bridged. For any Euclidean graph $G = (P, E)$ it holds that $G$ is a $t$-spanner if and only if all pairs of points $(a, b)$, $a, b \in P$, with $|ab| \leq \lambda$ have a $t$-path in $G$.
\end{lemma}
\begin{proof}
Follows by induction over all pairs of points $(a, b)$ with $|ab|$ ascending and earlier observations.
\end{proof}

We now develop a sufficient geometric condition for bridging pairs of points.

\begin{figure}[!h]
 \centering
 \vspace{-.5\baselineskip}
 \includegraphics{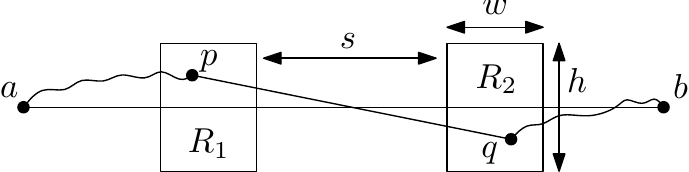}
 \caption{$(p, q)$ bridges $(a, b)$}
 \label{fig:twopoints}
 \vspace{-1.5\baselineskip}
\end{figure}

\begin{lemma}
\label{lemma:nicebox}
Suppose we are given points $a, b \in P$, rectangles $R_1$ and $R_2$ and $t > 1$, such that (as per Fig.~\ref{fig:twopoints}): $R_1$ and $R_2$ lie in between $a$ and $b$, have a side parallel to $ab$, have their centers on line segment $ab$, both have width $w$ and height $h$, are separated by $s \geq \frac{t+1}{t-1}h$ and $R_1$ lies closer to $a$ than $R_2$.

Then, for any $p, q \in P$ with $p$ lying in $R_1$ and $q$ lying
in $R_2$, $(p, q)$ bridges $(a, b)$.
\end{lemma}
\begin{proof}
To simplify the proof, we assume without loss of generality that $ab$ lies on the $x$-axis. For any $u, v \in P$, we denote the difference in $x$-coordinates of $u$ and $v$ as $d_x(u, v)$. We have $d_x(p,q) \geq s \geq \frac{t+1}{t-1}h$, so $h \leq \frac{t-1}{t+1}d_x(p,q)$, which leads to the lemma using the triangle inequality as follows:
\begin{align*}
	&t |ap|+|pq|+t |qb| \leq \\&t \left(d_x(a,p)+ \frac{1}{2}h\right) +( d_x(p,q) + h )+t \left( d_x(q,b)+ \frac{1}{2}h\right) \leq t d_x(a,b) = t|ab|
\end{align*}
\end{proof}

We now use Lemma~\ref{lemma:nicebox} to prove a stronger statement that we will use to prove the full version of Theorem~\ref{theorem:uniformbridged}. Let $a, p, q \in P$ be pairwise different points and let region $A \subseteq \Reals^2$ with $a, p, q \not \in A$. We say that the pair $(p, q)$ \emph{bridges} $(a, A)$ if for every point $b \in P$ with $b \in A$ we have that $(p, q)$ bridges $(a, b)$.

\begin{figure}
\includegraphics{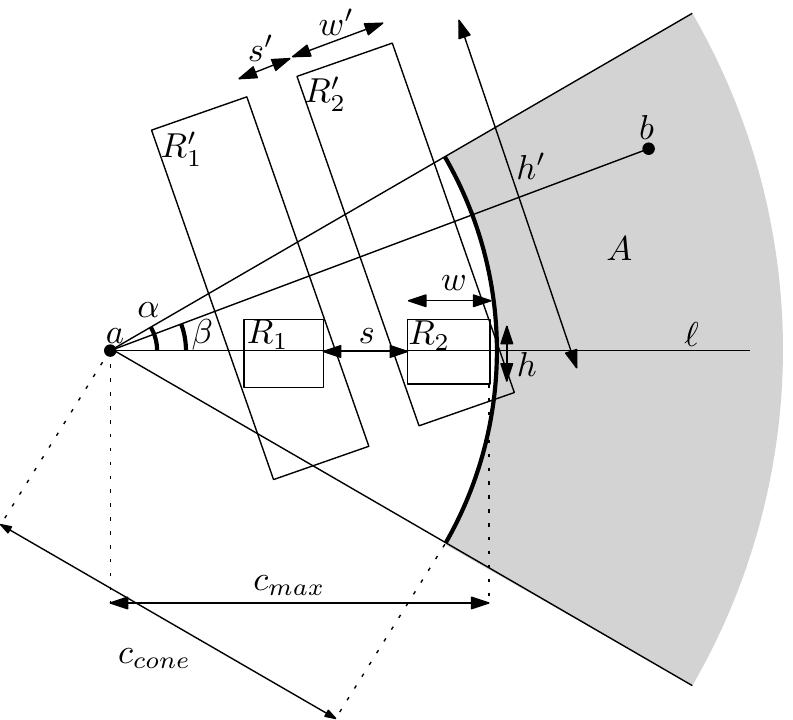}
\centering
\vspace{-.5\baselineskip}
 \caption{$R_1$ and $R_2$ are covered by $R_1'$ and $R_2'$, according to Lemma~\ref{lemma:nicebox}}
  \label{fig:boxinbox}
\vspace{-1em}
\end{figure}
\begin{lemma}
\label{lemma:conelemma}
Assume we are given $a \in P$, a line $\ell$ through $a$, an angle $\alpha \leq \pi/4$, a constant $c_{max}$, rectangles $R_1$ and $R_2$ and $t>1$, such that (as per Fig.~\ref{fig:boxinbox}): $R_1$ and $R_2$ have width $w$ and height $h$, are separated by $s$, have a side parallel to $\ell$, have their centers on $\ell$, $R_1$ lies between $a$ and $R_2$, $R_2$ lies at most $c_{max}$ away from $a$, $R_1$ lies at least $h/2$ away from $a$ and $s \geq \sqrt{2}\frac{t+1}{t-1}\left(2\sin(\alpha)c_{max} + h\right) + h$.

For the cone with apex $a$, angle $2\alpha$ and bisector $\ell$, we define $A$ as the area that is at least $c_{cone}=c_{max} + h/2$ away from $a$. Then for any $p, q \in P$ with $p$ lying in $R_1$ and $q$ lying in $R_2$, $(p, q)$ bridges $(a, A)$.
\end{lemma}

\begin{proof}
Let $b \in A$ and $b \in P$. We will prove that the rectangles $R_1$ and $R_2$ can be covered by rectangles $R_1'$ and $R_2'$ respectively, that meet all requirements of Lemma~ \ref{lemma:nicebox}, which therefore implies that the pair $(p,q)$ bridges $(a,b)$. The lemma then follows.

The rectangles $R_1'$ and $R_2'$ are chosen such that their centers lie on line segment $ab$, they lie in between $a$ and $b$ (this is where $c_{cone}=c_{max} + h/2$ is needed) and have at least one side parallel to $ab$. The rectangles are chosen to have equal width (= length of the size parallel to $ab$) $w'$ and height $h'$. Their position, height and width are chosen as the minimal values such that $R_1'$ contains $R_1$ and $R_2'$ contains $R_2$ (while maintaining the previous properties), as depicted in Fig.~\ref{fig:boxinbox}. Let $s'$ be the separation between $R_1'$ and $R_2'$ and let $\beta$ be the angle between $l$ and $ab$. Using basic geometry we can derive that:
\begin{equation*}
s' = \cos(\beta)s-h\sin(\beta).
\end{equation*}
The angle $\beta \leq \alpha$ is bounded by $\pi/4$. This implies that $\cos(\beta) \geq \frac{\sqrt{2}}{2}$ and $\sin(\beta) \leq \frac{\sqrt{2}}{2}$. We obtain the following lower bound on $s'$:
\begin{equation*}
s' \geq \frac{\sqrt{2}}{2}(s-h).
\end{equation*}
Substituting the lower bound assumed for the lemma and using that $\sin \alpha \geq \sin \beta$ we have:
\begin{equation*}
\label{eqn:sbound}
s' \geq \frac{t+1}{t-1}\left(2 \sin(\beta)c_{max}+ h\right).
\end{equation*}
We bound $h'/2$ by the distance from the center of the right side of $R_2$ to $ab$ plus the distance from this center to the corner of $R_2$:
\begin{equation*}
\label{eqn:hbound}
h' / 2 \leq \sin(\beta) c_{max} + \frac{h}{2},
\end{equation*}
Combining the bounds on $h'$ and $s'$ gives
\begin{equation*}
s' \geq \frac{t+1}{t-1}h'.
\end{equation*}
This proves that all requirements of Lemma~\ref{lemma:nicebox} hold. Hence $(p, q)$ bridges $(a,b)$.
\end{proof}

\safespace
\section{Uniform Point Sets} \label{section:uniformpoints}
\safespace

\begin{theorem} \label{theorem:uniformbridged}
There exists $c_t$ dependent only on $t$ such that for every $c > 0$, if $P$ is a set of points uniformly and independently distributed at random in a $\sqrt{n} \times \sqrt{n}$ square and $n$ is large enough, then with probability at least $1 - n^{-c}$, $P$ is locally-$(c \cdot c_t \log n)$-bridged.
\end{theorem}

We first give a high level overview of the proof followed by the complete proof.
We need to prove that every point in $P$ is locally-$(c \cdot c_t \log n)$-bridged simultaneously with high probability. We show that every point individually
is locally-$(c \cdot c_t \log n)$-bridged with sufficiently high probability that a simple union bound shows that it will happen to all points simultaneously with high probability. We use Lemma~\ref{lemma:conelemma} to achieve this. For ease of presentation, we assume $t$ is constant.

The rectangles in Lemma~\ref{lemma:conelemma} can be chosen to have a roughly constant chance of containing a point, and if we can fulfill the other requirements, the resulting pair of points bridges a relatively large part of $\Reals^2$. In fact, we need only $\lceil \pi / \alpha \rceil$ cones (we will end up picking $\alpha = O(1/\log n))$ to cover the area we wish to cover, as depicted in Fig.~\ref{fig:circlecones}. We show the likely existence of a pair of mandatory points that bridges a single cone and use a union bound to show such pairs are likely to exist for all cones simultaneously.

\begin{wrapfigure}[15]{r}{0.35\textwidth} %
\raggedleft
\vspace{-1.5\baselineskip}
 \includegraphics{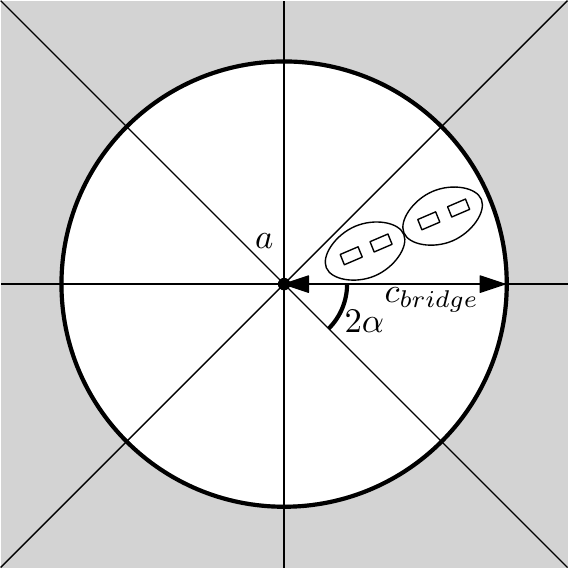}
\vspace{-1.5\baselineskip}
  \captionof{figure}{Covering the plane with cones}
  \label{fig:circlecones}
\vspace{-1em}
\end{wrapfigure}
We will place $O(\log n)$ pairs of rectangles in every cone as depicted in Fig.~\ref{fig:circlecones}. If any pair of boxes ends up containing a point per box, these two points will satisfy the requirements for Lemma~\ref{lemma:conelemma}. We just need this pair of points to be mandatory, and therefore consider an ellipse around such a pair of boxes (defined in terms of the boxes, not the points, for easy analysis), such that if this ellipse is empty apart from these two points, these points must be mandatory. Using a careful analysis, the chance that a pair of boxes contains one point per box and the ellipse contains no more points (an event we will call a `success') is at least some constant $p$ (dependent only on $t$). We need only one success per cone and the events are nearly independent (the ellipses do not overlap), so the chance that we get at least one success is at least (roughly) $1-p^{O(\log n)} = 1-n^{-O(f(t))}$, which then shows the theorem.

\pagebreak
We now give the full proof of Theorem~\ref{theorem:uniformbridged}.
\begin{proof}
Note that we will often introduce a constant (say, the height $h$ of $R_1$), give it a value (say $h=1$) but still refer to the name of the variable later for clarity (so $h$ instead of $1$).

\subsubsection{Positioning the Cones}

Let $c$ be given as per the theorem. Let $k=4 \exp\left(\frac{604 t^{7/2}}{(t-1)^{3/2}}\right)(c + 14)$. Let $c_{max} = 12 (k \log (n)+1) \frac{t^2}{t-1}$. Let $m = \left\lceil\pi / \arcsin\left(\frac{1}{2 c_{max}}\right)\right\rceil$. We partition the circle with radius $c_{cone} = c_{max} + 1/2$ around every point $a$ into $m \log n$ cones, as depicted in Fig.~\ref{fig:circlecones}. We want the area in every cone within the circle to fall entirely within the square. If $a$ lies near the edge of the square, this may not always be the case, so for these cones we either remove them or rotate them slightly around $a$ as follows.

We only aim to prove that $a$ is $c_{bridge} = \sqrt{2} \cdot c_{cone}$-bridged (and not $c_{cone}$-bridged), so we remove all cones whose area further than $c_{bridge}$ from $a$ lies outside the square in its entirety. For all other cones, if a point lies sufficiently far from a corner, it is easy to see we can just rotate the cone a bit so that the area closer than $c_{cone}$ from $a$ lies entirely within the square while the area that is further than $c_{bridge}$ from $a$ but still within the square is the same for the original and the rotated cone.

The only potential problem occurs when rotating a cone makes it end up outside the square if $a$ lies near a corner of the square. However, this means that the area of the cone further than $c_{bridge}$ away from $a$ but still within the square contains the corner of the square, but it is easily seen that this means that at least one of the edges of the square is more than $c_{cone}$ away from $a$, so this is never a problem. Note that rotated cones may overlap other cones, causing dependency issues that we will deal with later.

\begin{figure}
\centering
\vspace{-1.0em}
  \includegraphics{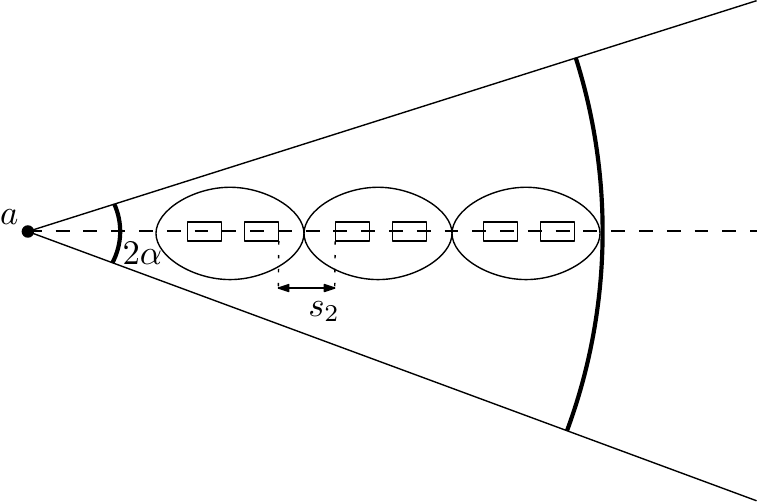}
  \captionof{figure}{Rectangle configuration in a cone}
  \label{fig:coneboxes}
\vspace{-3em}
\end{figure}

\subsubsection{Boxes in Cones}

We place $k \log n$ rectangles $R_1$ and $R_2$ in every cone as per Lemma~\ref{lemma:conelemma}, as depicted in Fig.~\ref{fig:coneboxes}. Every rectangle has width $w = 1$, height $h = 1$, and $R_1$ and $R_2$ are placed $s=h + 2\sqrt{2}\frac{t+1}{t-1}$ apart. The rectangles are aligned with the bisector $\ell$ of the cones. Neighboring pairs of rectangles are placed $s_2$ apart. Let $A = w \cdot h = 1$ be the area of $R_1$ and $R_2$.

We surround the rectangles by an ellipse $E$ with focii $d$ and $e$ and eccentricity $t$ as follows. The centers of $R_1$ and $R_2$ lie on $de$ and $d$ and $e$ are placed at a distance $X = \frac{h + 2w + t h}{2(t - 1)} = \frac{3+t}{2t-2}$ from $R_1$ and $R_2$ respectively. We now note that if $a \in R_1$ and $b \in R_2$, then any point $c$ lying in the ellipse with focii $a$ and $b$ and eccentricity $t$ also lies in $E$ as follows: from $|ac| + |cb| \leq t |ab| \leq t (h + 2w + s)$ we conclude $|dc| + |ce| \leq |da| + |ac| + |cb| + |be| \leq 2X + h + 2w + t (h + 2w + s) = t (2X + 2w + s) = t|de|$, and so $c \in E$. Properties of ellipses and algebraic simplification gives us
\begin{equation*}
|E| = \pi \frac{t \sqrt{t+1} (\sqrt{2} t+3 t+\sqrt{2})^2}{(t-1)^{3/2}} \leq \pi \frac{t \sqrt{t+t} (\sqrt{2} t+3 t+\sqrt{2}t)^2}{(t-1)^{3/2}} \leq \frac{151 t^{7/2}}{(t-1)^{3/2}}
\end{equation*}

If at least one point ends up in $R_1$, and at least one point ends up in $R_2$ and no other point ends up in $E$ (making the pair of points mandatory), then we say that this pair of rectangles is a \emph{success}. Let $\alpha = \arcsin\left(\frac{1}{2 c_{max}}\right)$, then the cones have angle $\leq 2\alpha$, which implies that $s \geq \sqrt{2}\frac{t+1}{t-1}\left(2 \sin(\alpha)c_{max} + h\right) + h$. The pair of points corresponding to a success would therefore fit the conditions of Lemma~\ref{lemma:conelemma} and would therefore bridge the cone we are considering. The Lemma requires that the angle of the cones is at most $\pi/4$, which follows from $m \geq 8$, which follows from $m \geq \frac{12 (k \log n+1) t^2}{t-1} - 1$, which follows from $\frac{t^2}{t-1}>1$ and $k \log n \geq 1$. We will show that we will have at least one success for every cone simultaneously with high probability.

We first consider the final condition that needs to be met: the first box must lie far enough away from the origin point so the ellipse around it lies entirely within the cone. The ellipse has minor axis $\frac{2 \sqrt{t+1} ((3+\sqrt{2}) t+\sqrt{2})}{\sqrt(t-1)} \leq 17 \frac{t^{3/2}}{\sqrt{t-1}}$ and the cone is therefore wide enough for this at $\sin(\alpha) \frac{17/2 t^{3/2}}{\sqrt{t-1}} = \frac{17}{4 c_{max}} \frac{t^{3/2}}{\sqrt{t-1}}$. The major axis of the ellipse is $\frac{2 t (\sqrt{2} t+3 t+\sqrt{2})}{t-1} \leq 12 \frac{t^2}{t-1}$, so to accommodate $k \log n$ ellipses, we need $c_{max} \geq \frac{17}{4 c_{max}} \frac{t^{3/2}}{\sqrt{t-1}} + 12 k \log (n) \frac{t^2}{t-1}$, which holds (after simplification).

\subsubsection{Probability of Success}

Let $p$ be the probability of success for a rectangle. Although the rectangles and ellipses do not overlap, the probability distributions for the rectangles are not independent, for if a pair of rectangles is not a success, then we learn something about the point sets: the points either avoid $R_1$, avoid $R_2$ or end up in $E$ too often or in $E \setminus (R_1 \cup R_2)$. We can therefore not immediately bound the chance that no pair of rectangles in a cone succeeds $p'$ by $(1-p)^{k \log n}$. If we keep the dependencies in mind, we can however get a bound that is almost as strong. The chance that a point ends up in an area $S$ may be higher than $\frac{S}{n}$, up to $\frac{S}{n - |E| k \log n}$, and the number of points we do not yet know the exact location of may be less than $n$.

We bound the chance $p_e$ that more than $k \log (n) (\sqrt{2 |E|} + |E|)$ points end up in the union of the ellipses (of a single cone): if we assume this happens, there are at least $n - k \log (n) (\sqrt{2 |E|} + |E|)) \geq n - 3 k |E| \log (n)$ points that we do not yet know the exact location of. We can bound $p_e$ by a binomial distribution with $p_C = k \log n \frac{|E|}{n}$, $n_C=n$ and $k_C = k \log (n) (\sqrt{2 |E|} + |E|)$ and a Chernoff bound: this gives us (after filling in and simplifying) $$p_e \leq \exp\left(-\frac{1}{2p_C}\frac{(n_Cp_C-k_C)^2}{n_C}\right) \leq n^{-k}$$

We will now bound the chance $p''$ that a pair of rectangles is a success assuming that no more than $3 k |E| \log (n)$ points end up in the union of the ellipses, and assuming that for any number of the other pairs of rectangles, we are given that they are either a success or not (thus allowing us to use the bound by $(1-p'')^{k \log n}$ later). For a success, we need two points to hit the rectangles (two factors $\frac{A}{n}$), no other points hit the rectangle (a factor $(1 - 2\frac{|E|}{n-3 k |E| \log (n)})^n$), with an additional factor $(n - 3 k |E| \log (n))^2$ because there are at least that many ways of picking the first two points.

\begin{align*}
p'' &\geq (n - 3 k \log (n) |E|)^2 \left(\frac{A}{n}\right)^2 \left(1 - 2\frac{|E|}{n-3 k \log (n) |E|}\right)^n\\
&\geq A^2 \left(1 - 3 k |E| \frac{\log n}{n}\right)^2 \exp\left(- 2\frac{|E|}{1-3 k |E| \log (n)/n}\right)\\
&= \left(1 - 3 k |E| \frac{\log n}{n}\right)^2 \exp\left(- 2\frac{|E|}{1-3 k |E| \log (n)/n}\right)\\
&:= b\\
\end{align*}

Note that $1 - 3 k |E| \frac{\log n}{n}$ goes to 0 as $n$ increases, so $b \approx \exp(-2|E|)$. Using that $b < 1$ we conclude that the chance $p'$ that no pair of rectangles in a cone succeeds assuming that no more than $3 k |E| \log (n)$ points end up in the union of the ellipses is at most
\begin{equation*}
(1-b)^{k \log n} = n^{k \log (1-b)} \leq n^{- k b} = n^{- k \left(1 - 3 k |E| \frac{\log n}{n}\right)^2 \exp\left(- 2 |E| / \left(1-3 k |E| \frac{\log n}{n}\right)\right)}
\end{equation*}

\subsubsection{Conclusion}

We now use a union bound to bound the chance that some cone either ends up without successes, or has too many points inside its ellipses. There are $m \log n$ cones per point and $n$ points, so this chance is at most

\begin{align*}
&m n \log (n) \left(n^{-k \log n} + n^{- k \left(1 - 3 k |E| \frac{\log n}{n}\right)^2 \exp\left(- 2 |E| / \left(1-3 k |E| \frac{\log n}{n}\right)\right)}\right)\\
&\leq \left\lceil\pi / \arcsin\left(\frac{1}{2 c_{max}}\right)\right\rceil n \log (n) \left(n^{-k \log n} + n^{- k \left(1 - 3 k |E| \frac{\log n}{n}\right)^2 \exp\left(- 2 |E| / \left(1-3 k |E| \frac{\log n}{n}\right)\right)}\right)\\
&\leq \left(1 + \frac{24 \pi (k \log (n)+1) t^2}{t-1}\right) n \log (n) \left(n^{-k \log n} + n^{- k \left(1 - 3 k |E| \frac{\log n}{n}\right)^2 \exp\left(- 2 |E| / \left(1-3 k |E| \frac{\log n}{n}\right)\right)}\right)\\
&\leq n^{1 + \frac{\log \log n}{\log n} + \frac{\log\left(1 + \frac{24 \pi (k \log (n)+1) t^2}{t-1}\right)}{\log n}} \left(n^{-k \log n} + n^{- k \left(1 - 3 k |E| \frac{\log n}{n}\right)^2 \exp\left(- 2 |E| / \left(1-3 k |E| \frac{\log n}{n}\right)\right)}\right)\\
\end{align*}

We wish for the above chance to become $\leq n^{-c}$. Noting that $n^a + n^b \leq 2 n^{\max(a, b)}$\\$= n^{\max\left(a + \frac{\log(2)}{\log(n)}, b + \frac{\log(2)}{\log(n)}\right)}$, we will bound both exponents in the above chance by $-c-\log(2)/\log(n)$. We assume that $n > 906 k \frac{t^{7/2}}{(t-1)^{3/2}}$, which makes $1-3 k |E| \frac{\log (n)}{n} > 1/2$. We will use $t < n$ and $k < n$ which follow from $n > 906 k \frac{t^{7/2}}{(t-1)^{3/2}}$, as well as $\log(n)-1 > \frac{\log(n)}{2}$ from $n>8$.

\begin{align*}
-k \log n + 1 + \frac{\log \log n}{\log n} + \frac{\log\left(1 + \frac{24 \pi (k \log (n)+1) t^2}{t-1}\right)}{\log n} + \frac{\log 2}{\log n} \leq -c\\
-k (\log (n) - 1) + 3 + \frac{\log \log n}{\log n} + \log (\log (n)) + \frac{\log\left(\frac{24 \pi t^2}{t-1}\right)}{\log n} + \frac{\log 2}{\log n} \leq -c\\
k \geq \frac{c + 3 + \frac{\log \log n}{\log n} + \log(\log(n))+\frac{\log\left(\frac{24 \pi t^2}{t-1}\right)}{\log n} + \frac{\log 2}{\log n}}{\log (n) - 1}\\
k \geq 2 \frac{c + 3 + 2\log(\log(n))+7\frac{\log t}{\log n}}{\log n}\\
k >= 2 \frac{c + 10 + 2\log(\log(n))}{\log n}\\
k >= \frac{2c + 20}{\log n} + 4\\
\end{align*}

This bound holds by our definition of $k$. We now turn to the other exponent.

\begin{align*}
- k \left(1 - 3 k |E| \frac{\log n}{n}\right)^2 \exp\left(\frac{- 2|E|}{1 - 3 k |E| \frac{\log n}{n}}\right) + 1\\+ \frac{\log \log n}{\log n} + \frac{\log\left(1 + \frac{24 \pi (k \log (n)+1) t^2}{t-1}\right)}{\log n} + \frac{\log 2}{\log n} &\leq -c\\
- \frac{k}{4} \exp(- 4 |E|) + 3 + 2\frac{\log \log n}{\log n}\\+ \frac{\log k}{\log n} + \frac{\log\left(\frac{24 \pi t^2}{t-1}\right)}{\log n} + \frac{\log 2}{\log n} &\leq -c\\
\frac{k}{4} \exp(- 4 |E|) - \frac{\log n}{\log n} &\geq c + 3 + 2\frac{\log \log n}{\log n} + \frac{\log\left(\frac{24 \pi t^2}{t-1}\right)}{\log n} + \frac{\log 2}{\log n}\\
k &\geq 4 \exp(4 |E|)(c + 4 + 2\frac{\log \log n}{\log n}\\&+ \frac{\log\left(\frac{24 \pi t^2}{t-1}\right)}{\log n} + \frac{\log 2}{\log n})\\
k &\geq 4 \exp\left(\frac{604 t^{7/2}}{(t-1)^{3/2}}\right)(c + 14)
\end{align*}

This bound also holds by our definition of $k$. We note that $k = O\left(c e^{\frac{t^{7/2}}{(t-1)^{3/2}}}\right)$ and so $c_{bridge} = O\left(c e^{\frac{t^{7/2}}{(t-1)^{3/2}}} \frac{t^2}{t-1} \log n\right)$ and the theorem follows. 
\end{proof}

\safespace
\section{Algorithms} \label{section:consequences}
\safespace

We first introduce three tools used in the results below. Let $c$ and $c_t$ be as in Theorem~\ref{theorem:uniformbridged} throughout this section. The first is that we can divide the input into a $\frac{\sqrt{n}}{c \cdot c_t \log n} \times \frac{\sqrt{n}}{c \cdot c_t \log n}$ grid in $O(n \log n)$ time, with every cell containing in expectation $O((c \cdot c_t \log n)^2)$ points.

The second tool is the `local' Dijkstra algorithm. It determines for all points at most $\lambda$ away from a source point $s$ whether it has a $t$-path to $s$ and if so, their network distance. It differs from the standard Dijkstra algorithm in that it only adds the points to the queue at most $\lambda t$ away from the source $s$ by considering the points lying in cells at most $\lambda t$ away from $s$, and only considers the edges $E_s$ that have such a point as either endpoint. Using the grid this can be done in $O((\lambda^2 + |E_s|) \log \lambda)$ expected time.

The third tool is called \emph{path-hyperbola}. It is an area given by an origin point $u \in P$, a focus $v \in P$ and an edge set $E$, and is defined as $PH(u, v, E) = \{ a \in \Reals^2 \mid \delta_{(P, E)}(u, v) + t \cdot |va| \leq t \cdot |ua| \}$. Obviously, if $(p, q)$ bridges $(a, b)$, then $b \in PH(a, q, E)$ for every edge set $E$ with $t$-paths for pairs of points $(u, v)$ with $|uv| \leq |ab|$, making path-hyperbola at least as powerful as bridging points for guaranteeing $t$-paths.

If we perform a local Dijkstra on $s$, we find a set of network distances that induce a set of path-hyperbola. If $s$ is locally-$\lambda$-bridged, the union of path-hyperbola will be a superset of the area more than $\lambda$ away from $s$, guaranteeing $t$-paths to all other points. This union can be computed in $O(\lambda^2 \log \lambda)$ expected time: using polar coordinates, the union corresponds to a lower envelope. Since the hyperbolas pairwise intersect at most twice, this envelope has linear complexity and can be computed in $O(n \log n)$ time~\cite{At85,sharir1995cup:dsbook}. We can therefore use this to test in $O(\lambda^2 \log \lambda)$ expected time whether $s$ has a $t$-path to all other points: if the local Dijkstra finds only $t$-paths but $s$ is not locally-$\lambda$-bridged, we can perform a normal Dijkstra without affecting the expected running time.

\subsection{Testing $t$-spanners}

The first application of Theorem~\ref{theorem:uniformbridged} and our tools is a faster algorithm to test if a Euclidean graph is a $t$-spanner on uniformly distributed point sets: we simply run the procedure from the previous section on every point. To the best of our knowledge, this leads to the first subquadratic algorithm for this problem on any interesting class of point sets not making assumptions on $E$. 

\begin{theorem} \label{theorem:testing}
There is an algorithm that, given a point set $P$ whose points are uniformly distributed in a $\sqrt{n} \times \sqrt{n}$ square and a Euclidean graph $E$ on $P$, checks if $E$ is a $t$-spanner using $O((n + |E|) (c_t \log n)^2 \log (c_t \log n))$ expected time, where $c_t$ is a constant dependent only on $t$.
\end{theorem}
\begin{proof}
Applying our three tools with $\lambda = c \cdot c_t \log n$ almost immediately gives us the desired result: we run a local Dijksta for every point, maintaining the union of the path hyperbola. If we find any pair of points without $t$-path, we return that the input is not a $t$-spanner. If some union of path-hyperbola for a point $s$ does not cover the area more than $\lambda$ away from $s$, we perform a $O(n^2 \log n)$ test for $t$-spannerness, and otherwise we return that the input is a $t$-spanner, which happens with high probability by Theorem~\ref{theorem:uniformbridged}. This algorithm therefore uses $O((n + |E|) (c_t \log n)^2 \log (c_t \log n))$ expected time.
\end{proof}

\subsection{Greedy Spanner}

\begin{algorithm}{GreedySpannerOriginal}[V, t]{\label{algo:greedyorig}}
  $E \gets \emptyset$
  \\ \qfor every pair of distinct points $(u, v)$ in ascending order of $|uv|$
  \\ \qdo \qif $\delta_{(V, E)}(u, v) > t \cdot |uv|$ \label{line:test}
       \\ \qthen add $(u, v)$ to $E$
          \qendif
     \qendfor
  \\ \qreturn $E$
\end{algorithm}

Consider the original algorithm above as introduced in~\cite{Keil:1988:ACE:61764.61787}.
The graph returned by this algorithm is called the \emph{greedy spanner} on $V$ for $t$ and it is obviously a $t$-spanner, but the algorithm has a $O(n^3 \log n)$ running time. 

\begin{lemma} \label{lemma:greedyshortedges}
If $P$ is $\lambda$-bridged, then the greedy spanner on $P$ does not have edges longer than $\lambda$.
\end{lemma}
\begin{proof}
After ensuring $t$-paths for all $(u, v)$ with $|uv| \leq \lambda$ the algorithm will not add more edges as all $(u, v)$ with $|uv| > \lambda$ have $t$-paths by Lemma~\ref{lemma:bridgedness}.
\end{proof}

We can combine Lemma~\ref{lemma:greedyshortedges} with Theorem~\ref{theorem:uniformbridged} to quickly compute the greedy spanner on uniform point sets. We first give a preliminary algorithm which we then employ in two greedy spanner algorithms.

\begin{theorem} \label{theorem:greedybasic}
For every $\lambda > 0$, there is an algorithm that, given a point set $P$ whose points are uniformly distributed in a $\sqrt{n} \times \sqrt{n}$ square, computes in $O(n \log n + n \lambda^2 \log^2 \lambda)$ expected time the edges of the greedy spanner on $P$ for $t$ of length at most $\lambda$.
\end{theorem}
\begin{proof}
We use the algorithm introduced in~\cite{AlewijnseBBB13} (we omit an explanation of the machinery introduced there), except we keep Lemma~\ref{lemma:greedyshortedges} in mind and use our local Dijkstra instead of a normal Dijkstra and only consider well-separated pairs $\{ A_i, B_i \}$ with $\min(A_i, B_i) \leq \lambda$.

Using the analysis in~\cite{AlewijnseBBB13} and using that the greedy spanner has degree $O(1)$, we conclude that if $m$ is the number of considered well-separated pairs, the running time of our modified algorithm is $O(n \log n + \lambda^2 \log \lambda \sum_{i=1}^m \min(|A_i|, |B_i|))$. We therefore need to bound \\$\sum_{i=1}^m \min(|A_i|, |B_i|) \leq \sum_{i=1}^m (|A_i| + |B_i|) = \sum_{a \in P} |\{ \{ A_i, B_i \} \mid a \in A_i \vee a \in B_i \}|$.

For any $l \in \Reals$, a point $p$ can only be in $O(1)$ well-separated pairs of length at most a constant factor higher or lower than $l$~\cite[Lemma 4.6.1]{Callahan95dealingwith}. We can therefore partition the well-separated pairs containing $p$ into $O(1)$-sized sets of similar length. As the minimal length per set differs by at least a constant factor, we conclude $|\{ \{ A_i, B_i \} \mid a \in A_i \vee a \in B_i \}| = O\left(\log \frac{\max_i\{l(\{ A_i, B_i \})\}}{\min_i\{l(\{ A_i, B_i \})\}}\right)$. This last expression is $O(\log \lambda)$ in expectation on uniform point sets, giving an expected running time of $O(n \log n + n \lambda^2 \log^2 \lambda)$.
\end{proof}

Note that we could have adapted the algorithm from~\cite{BoseCFMS2010}, but this algorithm sorts all potential edges, resulting in an expected $O(n \log n \lambda^2 \log \lambda)$ running time, which is slower when filling in $\lambda = O(\log n)$.

Combining Lemma~\ref{lemma:greedyshortedges}, Theorem~\ref{theorem:uniformbridged} and Theorem~\ref{theorem:greedybasic} (with $\lambda = c \cdot c_t \log n$) gives:

\begin{corollary}
	There is an algorithm that, given a point set $P$ whose points are uniformly distributed in a $\sqrt{n} \times \sqrt{n}$ square, computes in \\$O(n (c_t \log n)^2 \log^2 (c_t \log n))$ expected time a graph on $P$ which is with high probability the greedy $t$-spanner (with $c_t$ is a constant dependent only on $t$).
\end{corollary} 

\subsection{The Full Distribution-Sensitive Algorithm} \label{subsection:missingedges}

The algorithm from Theorem~\ref{theorem:greedybasic} is the first phase of our distribution sensitive algorithm. We now present the second and third phase that ensure that all long edges are also computed.

The second phase gathers path-hyperbola as described at the start of this section. We then consider the well-separated pairs that did not get considered in the first stage of the algorithm and try to prove for them that they will not produce a greedy spanner edge. For the remaining pairs, we employ the algorithm of~\cite{AlewijnseBBB13} in the third phase of our algorithm to find the remaining greedy 
spanner edges.

If for a point $u \in A_i$, the bounding box $B_i$ is covered by the union of path-hyperbola computed for $u$ (testing this takes $O(\log n)$ time), then we say $u$ is \emph{discounted} with respect to $\{ A_i, B_i \}$. If all $u \in A_i$ are discounted, then $\{ A_i, B_i \}$ will not contain a greedy spanner edge and we say $\{ A_i, B_i \}$ is \emph{discounted}. This can be computed in $O(\log n \sum_{i=1}^m (|A_i| + |B_i|)) = O(n \log n \log \lambda)$ expected time by an earlier argument.

We then perform the algorithm from~\cite{AlewijnseBBB13}, with small differences. We ignore pairs that have been discounted in the previous phase, and we do not perform a Dijkstra operation on points which have been discounted with respect to that pair as well. By Theorem~\ref{theorem:uniformbridged}, all pairs are discounted with high probability and hence this phase takes constant time in expectation on uniform point sets.

In practice, using a $\lambda$ lower than predicted by Theorem~\ref{theorem:uniformbridged} will suffice and be faster. From experiments we observe that $\lambda = \frac{\log n}{\sqrt[4]{t-1} \log \log n}$ is the `right' bound for the length of the longest edge in the greedy spanner. Using $1.1 \cdot \lambda$ the initial phase nearly always finds all edges, with the second phase usually discounting 99.7\% of the pairs and 95\% of the points in undiscounted pairs, with the second phase taking about 20\% of the time of the first. Using $1.5 \cdot \lambda$, all pairs are typically discounted.

\begin{theorem}
There is an algorithm that, given $t$ and a point set $P$ whose points are uniformly distributed in a $\sqrt{n} \times \sqrt{n}$ square, computes in $O(n (c_t \log n)^2 \log^2 (c_t \log n))$ expected time its greedy spanner, with $c_t$ a constant dependent only on $t$. The algorithm uses $O(n^2 \log^2 n)$ time on arbitrary $P$.
\end{theorem}

\safespace
\section{Experimental Results} \label{section:results}
\safespace

We have run our algorithm and WSPD-Greedy from~\cite{AlewijnseBBB13} on point sets whose size ranged from 500 to 128,000 points. The WSPD-Greedy algorithm has a running time comparable to the other (quadratic space) algorithms. Since running these on more then 10,000 points quickly becomes infeasible we did not include them in our experiments. For a detailed comparison between the major quadratic space algorithms and WSPD-Greedy we refer to~\cite{AlewijnseBBB13}. Note that we have verified that all our implemented algorithms give the same output.

Throughout this section we will refer to our algorithm as ``Bucketing'' in the graphs. We generated point sets according to several distributions. We have recorded space usage and running time (wall clock time). The results are averages over several runs where new point sets were generated each time. We included graphs for the uniform point set and for a clustered point set as these represent the best and worst cases respectively for our algorithm (with respect to our set of tests). To generate the clustered point set we used the same method as~\cite{AlewijnseBBB13}, that is, for $n$ points, it consists of $\sqrt{n}$ uniformly distributed point sets of $\sqrt{n}$ uniformly distributed points.

\subsection{Environment}

The algorithms have been implemented in C++. The random generator used was the Mersenne Twister PRNG -- we have used a C++ port by J. Bedaux of the C code by the designers of the algorithm,  M. Matsumoto and T. Nishimura. We have implemented all other necessary data structures and algorithms not already in the \verb|std| ourselves. The implementations do not use parallelism and run on a single thread.

Our experiments have been run on a server using an Intel Xeon E5530 CPU (2.40GHz) and 8GB (1600 MHz) RAM. It runs the Debian 7 OS and we compiled for 64 bits using G++ 4.7.2 with the -O3 option.

\subsection{Dependence on Instance Size}

We have compared running time and space usage of WSPD-Greedy and our algorithm for different values of $n$. We plotted the running time for $t=2$ on uniform and clustered points in Fig.~\ref{figure:runningtimeplot}. The space usage for both algorithms is linear but our algorithm uses a constant factor less space in practice.

\begin{figure}[h!]\centering\vspace{-2\baselineskip}
\includegraphics[width=8cm]{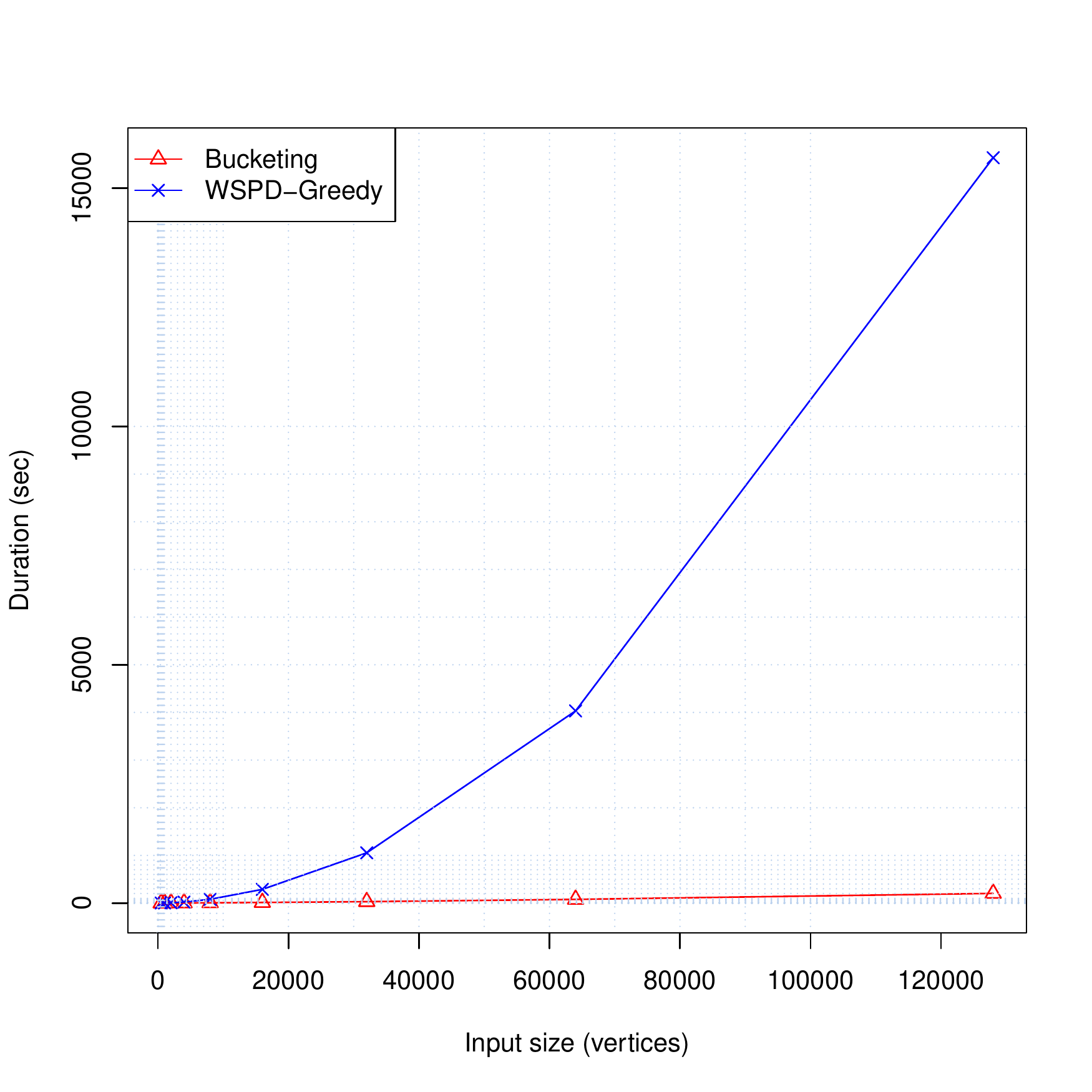}
\includegraphics[width=8cm]{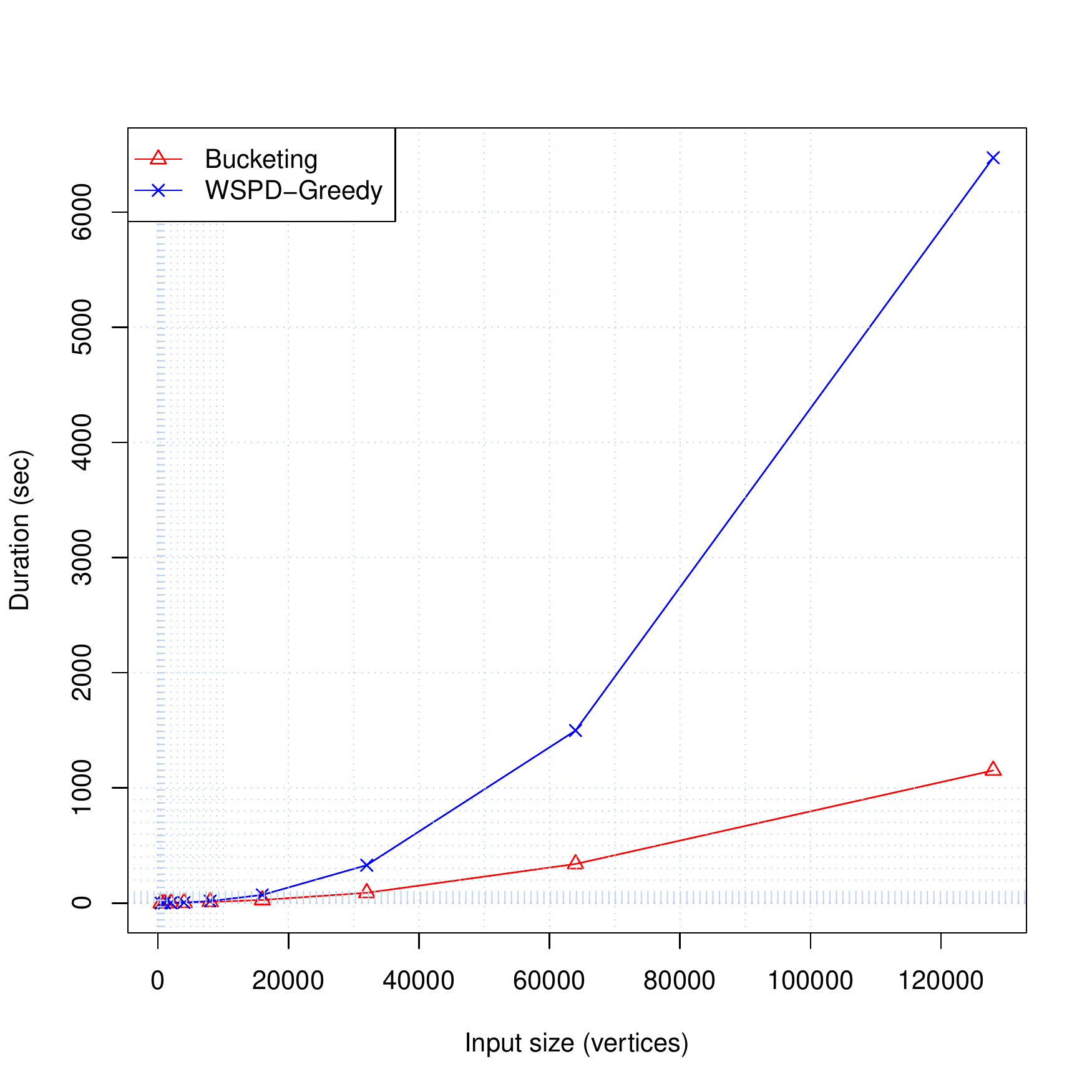}
\vspace{-1em}
\caption{The left plot shows the running time of our algorithm (Bucketing) and WSPD-Greedy for $t=2$ on variously sized uniformly distributed instances. The right plot shows the same for clustered instances.}
\vspace{-1.5em}
\label{figure:runningtimeplot}
\end{figure}
\begin{figure}[h!]\centering\vspace{-1\baselineskip}
\includegraphics[width=8cm]{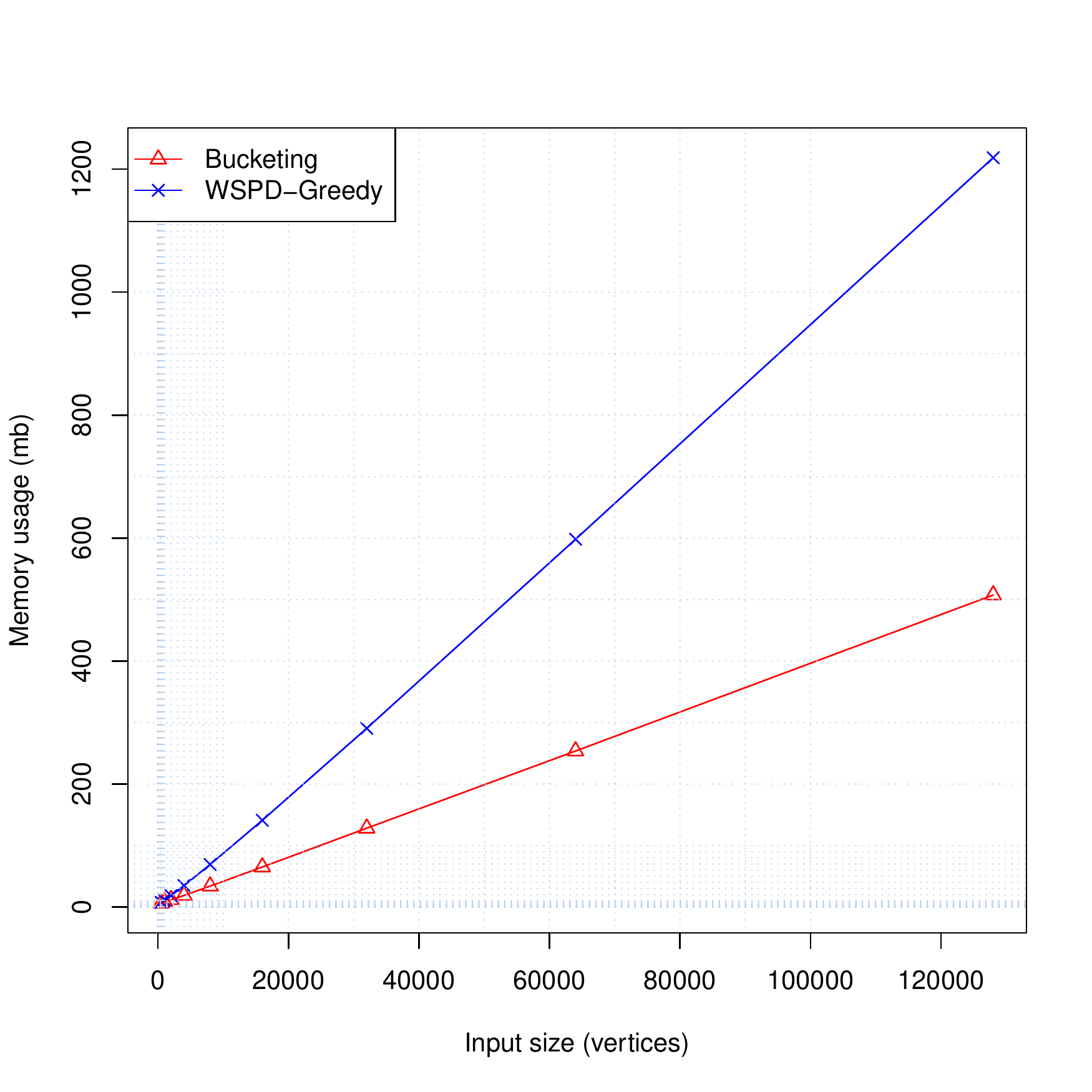}
\includegraphics[width=8cm]{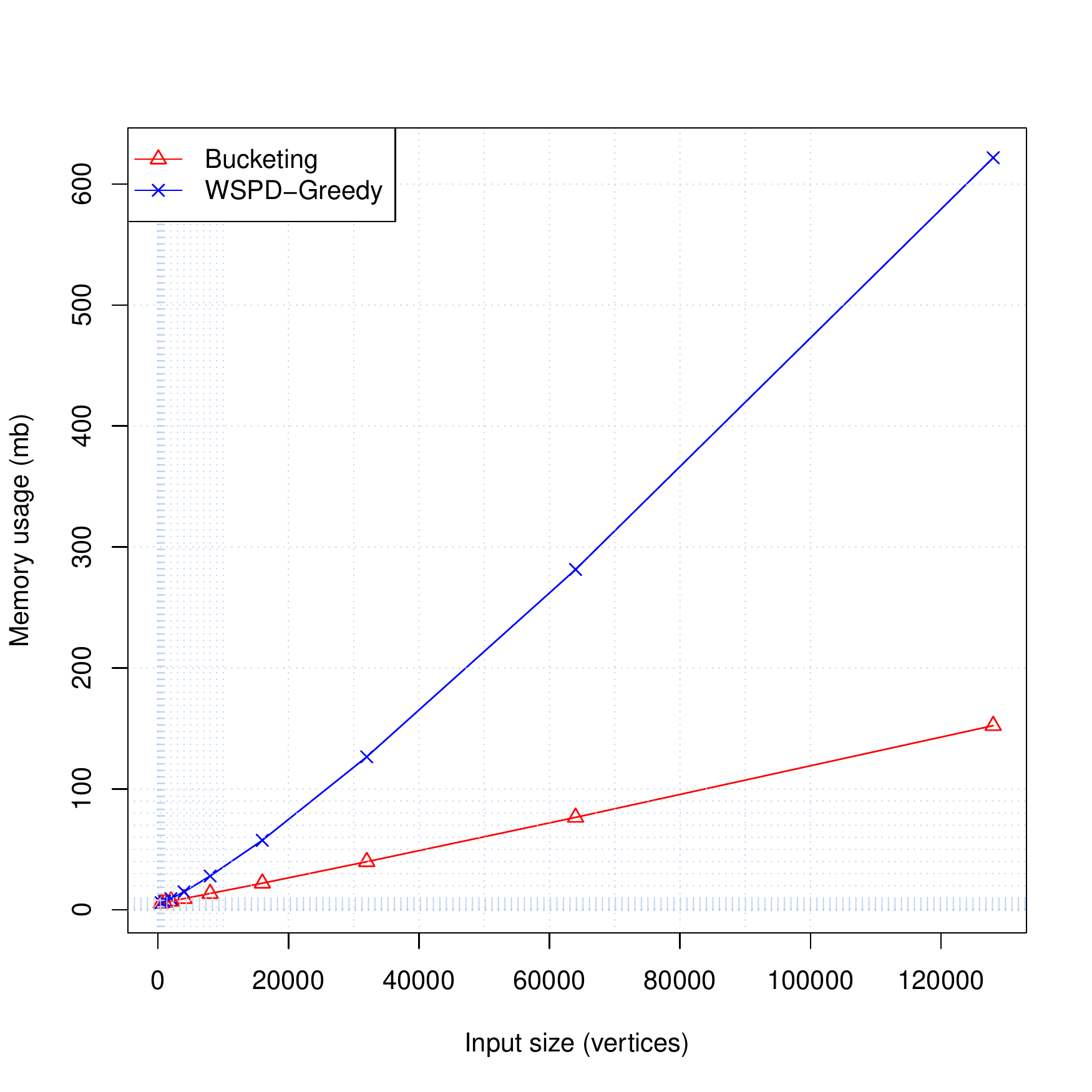}
\vspace{-1em}
\caption{The left plot shows the space usage of our algorithm (Bucketing) and WSPD-Greedy for $t=2$ on variously sized uniformly distributed instances. The right plot shows the the same for clustered instances.}
\vspace{-1\baselineskip}
\label{figure:squaresplot}
\end{figure}

The running time of our algorithm on uniformly distributed points is (nearly) linear making it a massive improvement over WSPD-Greedy. This allows us to calculate greedy spanners on such point sets in a matter of minutes where WSPD-Greedy would need hours or even days for bigger instances. 
 
The clustered point set is a bad case for our algorithm since the greedy spanner will contain a considerable amount of really large edges between clusters. Nevertheless, the algorithm still outperforms WSPD-Greedy by quite a margin. Our experiments on clustered data with smaller $t$ values (up to $t=1.1$) show that the performance of the algorithms gets more similar as $t$ decreases. On point sets drawn using a uniform or normal distribution our algorithm massively outperforms WSPD-Greedy for both small and large $t$. Additional plots for $t=1.1$ and for point sets using the normal distribution can be found in appendix \ref{app:plots}.
 
\subsection{Real Data}

Aside from generated instances we also experimented on some real point sets from the TSPLIB\footnote{http://comopt.ifi.uni-heidelberg.de/software/TSPLIB95/}. The performance of our algorithm on these sets seems to be close to the uniform point sets. Figure \ref{fig:real} shows two point sets and their greedy spanners. For the PCB the computation took on average about 2 seconds for $t=2$ and 11 seconds for $t=1.1$. The same computations using WSPD-Greedy took 12 and 203 seconds respectively. The bigger Germany instance took 21 and 147 seconds to compute using our algorithm while WSPD-Greedy needed 274 and 7,486 seconds for $t=2$ and $t=1.1$. This is a factor 50 improvement for the low $t$ case which reduces the computation time from hours to minutes.
 

\begin{figure}[h!]\centering
\vspace{-.5em}
\includegraphics[width=3.9cm]{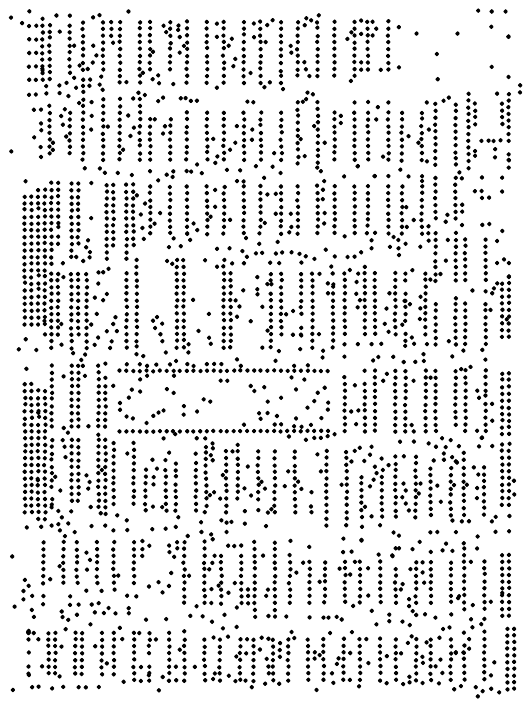}
\includegraphics[width=3.9cm]{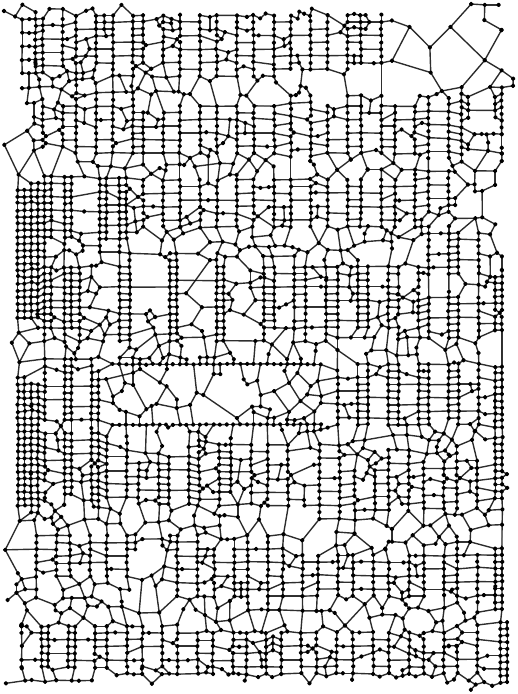} 
\includegraphics[width=3.9cm]{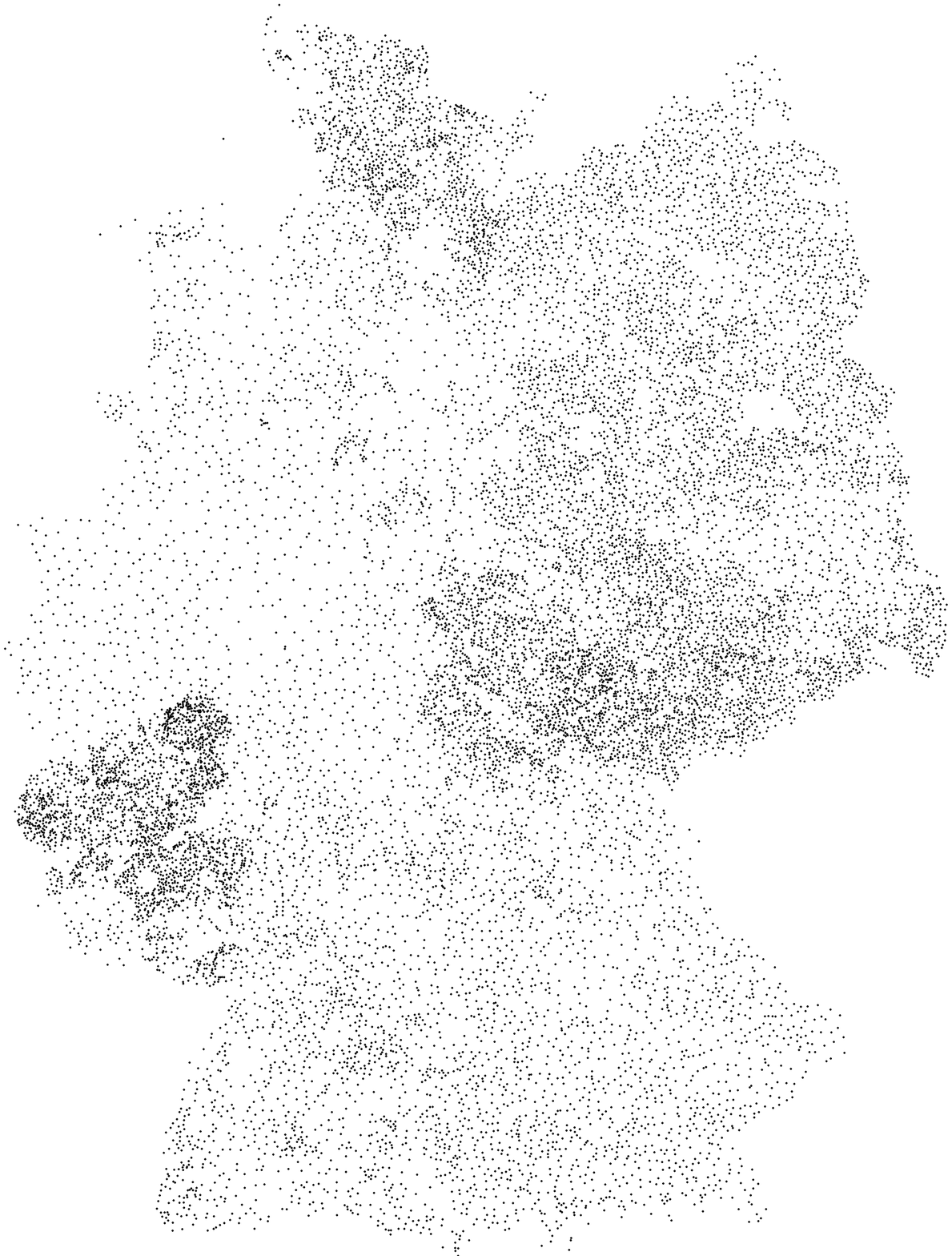}
\includegraphics[width=3.9cm]{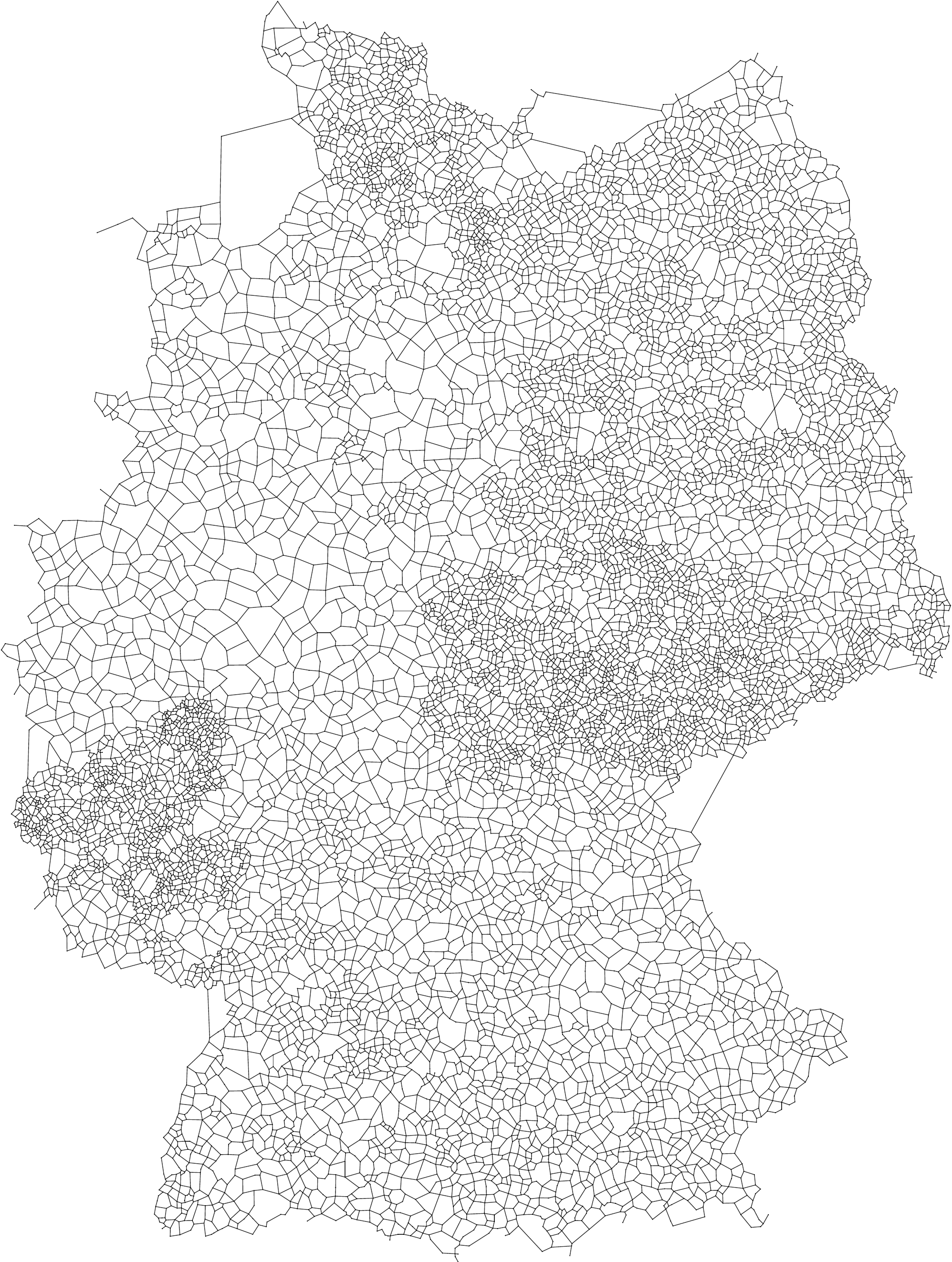}
\vspace{-.5em}
\caption{Real point sets from the TSPLIB and their greedy spanners using $t=2$. Left: A PCB instance of 3,038 points. Right: Cities in Germany, 15,112 points.}
\label{fig:real}
\end{figure}
 
\safespace
\section{Conclusion}
\safespace

We have introduced a distribution sensitive algorithm for computing the greedy spanner. Experiments show large improvements in both time and space for most data sets, while results are never worse than the state-of-the-art. The performance gap in many cases becomes even larger for lower $t$. To explain these results, we have analyzed the algorithm on uniformly distributed point sets.

To this end, we have introduced the concept of \emph{bridgedness} and have shown that point sets that are uniformly distributed in a $\sqrt{n} \times \sqrt{n}$ square are $O(\log n)$-bridged with high probability. This implies that `$t$-spannerness' is a `local' property on these point sets: a Euclidean graph is a $t$-spanner if and only if all pairs of `close-by' points have $t$-paths. This locality shows that our algorithm is near-linear on these point sets and yields a near-linear time algorithm for testing whether an edge set is a $t$-spanner on these point sets.

We leave open several questions that may be answered in future work. First, in our experiments, we have observed that the length of the longest edge of the greedy spanner on uniform point sets tends towards $\frac{\log n}{\sqrt[4]{t-1} \log \log n}$, leaving a gap with our upper bound; similarly, our bridgedness bound may also be improvable. 
Secondly, it would be interesting to see if our results generalize to higher dimensions. Lastly, there is still no general subquadratic time algorithm for the greedy spanner. Our algorithm could be considered a divide and conquer algorithm where the conquer step may be very slow, possibly susceptible to improvement.
\newpage
\bibliography{refs}
\bibliographystyle{abbrv}
\newpage
\appendix

\section*{\Large Appendix}

\section{Additional plots} \label{app:plots}

 \begin{figure}[h!]\centering
 \includegraphics[width=8cm]{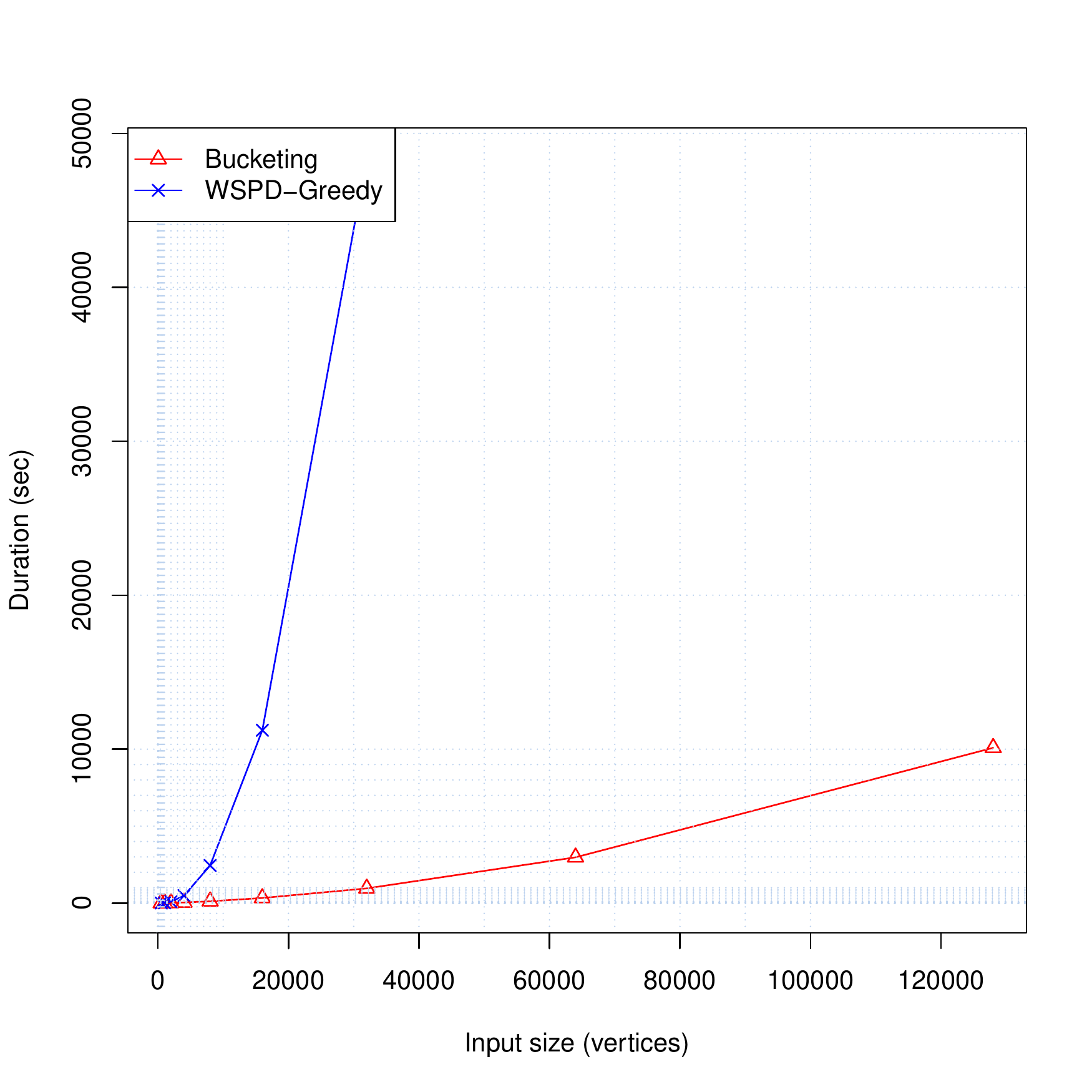}
 \includegraphics[width=8cm]{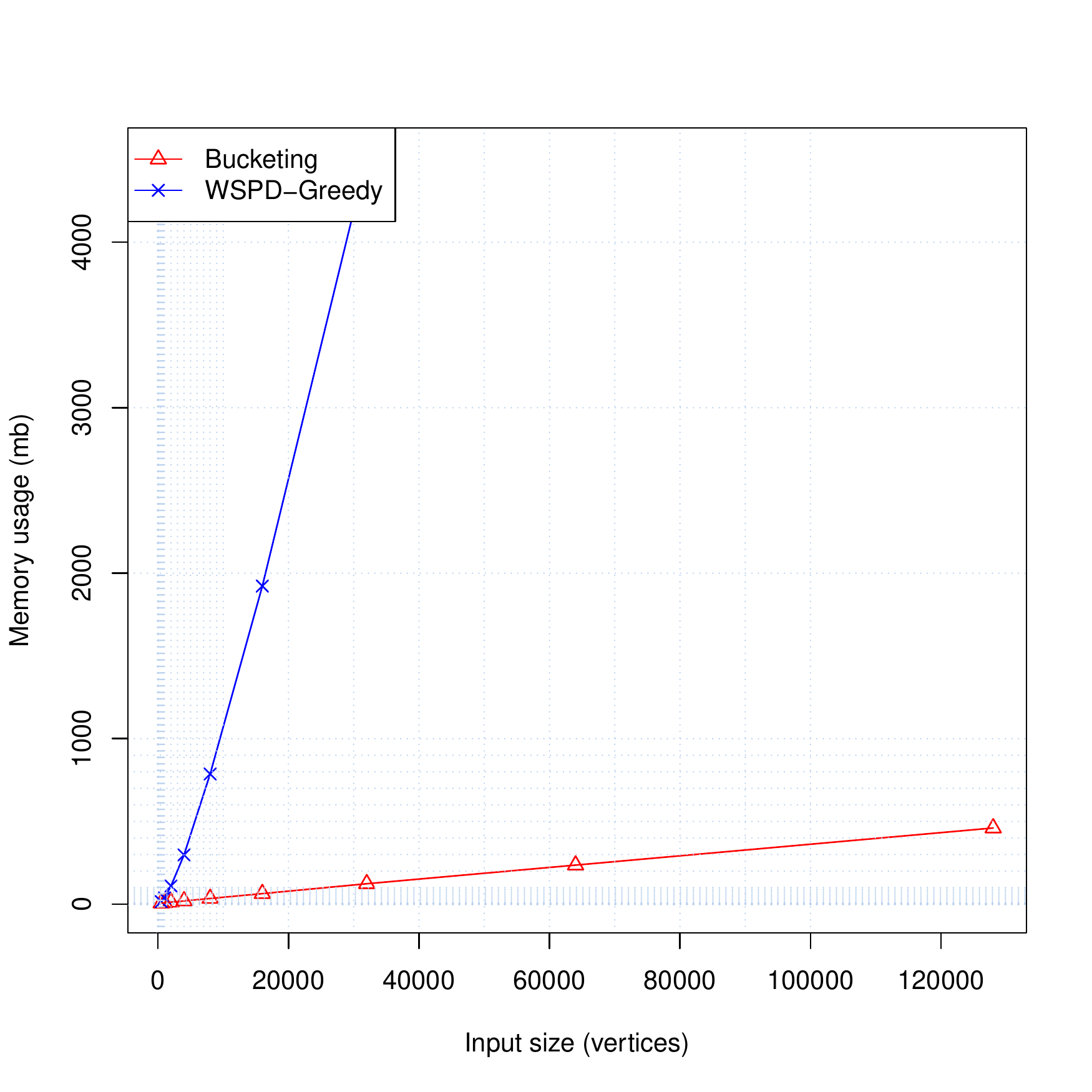}
 \vspace{-.5em}
 \caption{The left plot shows the running time of our algorithm (Bucketing) and WSPD-Greedy for $t=1.1$ on variously sized point sets generated with a normal distribution. The right plot shows the memory usage on the same data}
 \vspace{-1.5em}
 \label{figure:normaltplot}
 \end{figure}

 \begin{figure}[h!]\centering
 \includegraphics[width=8cm]{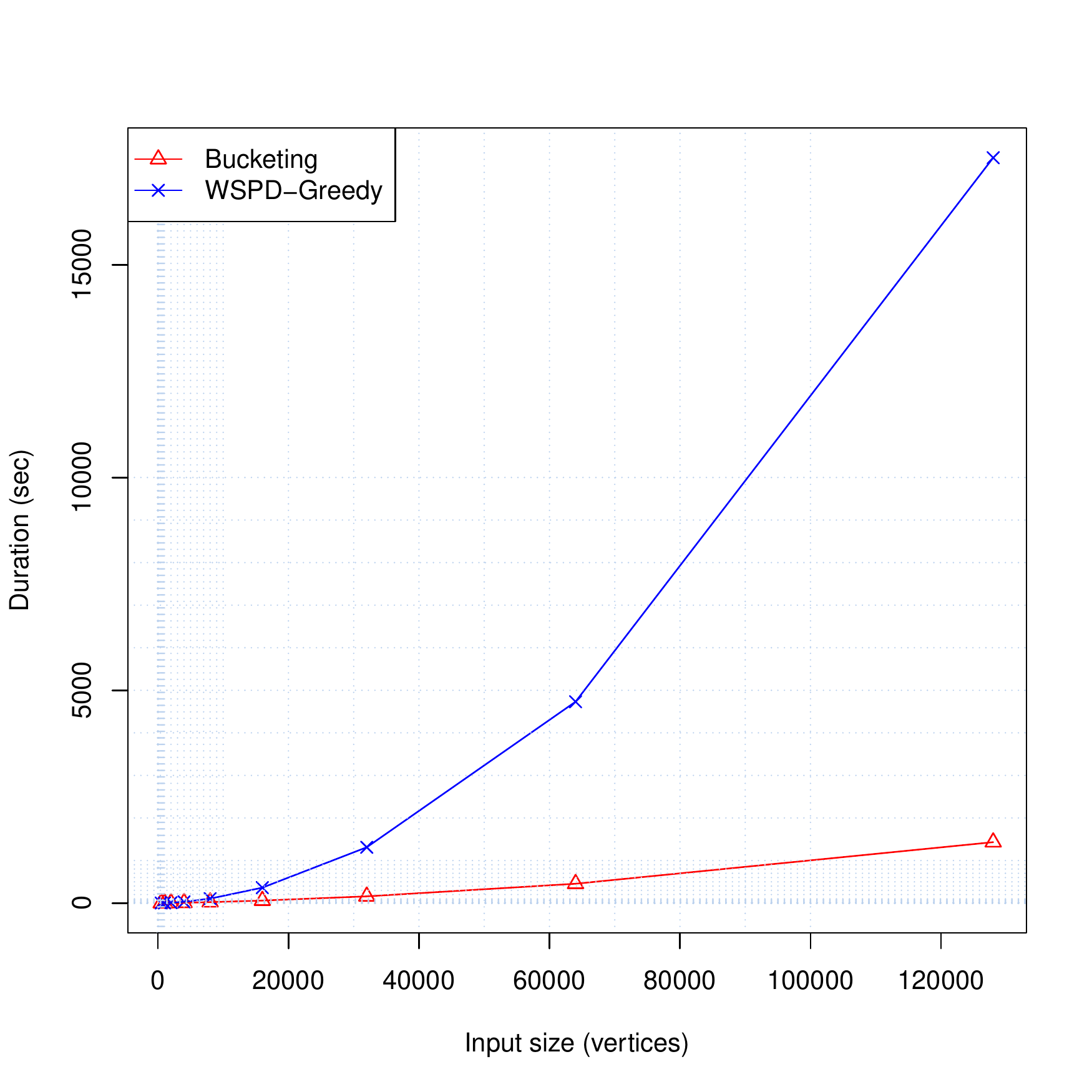}
 \includegraphics[width=8cm]{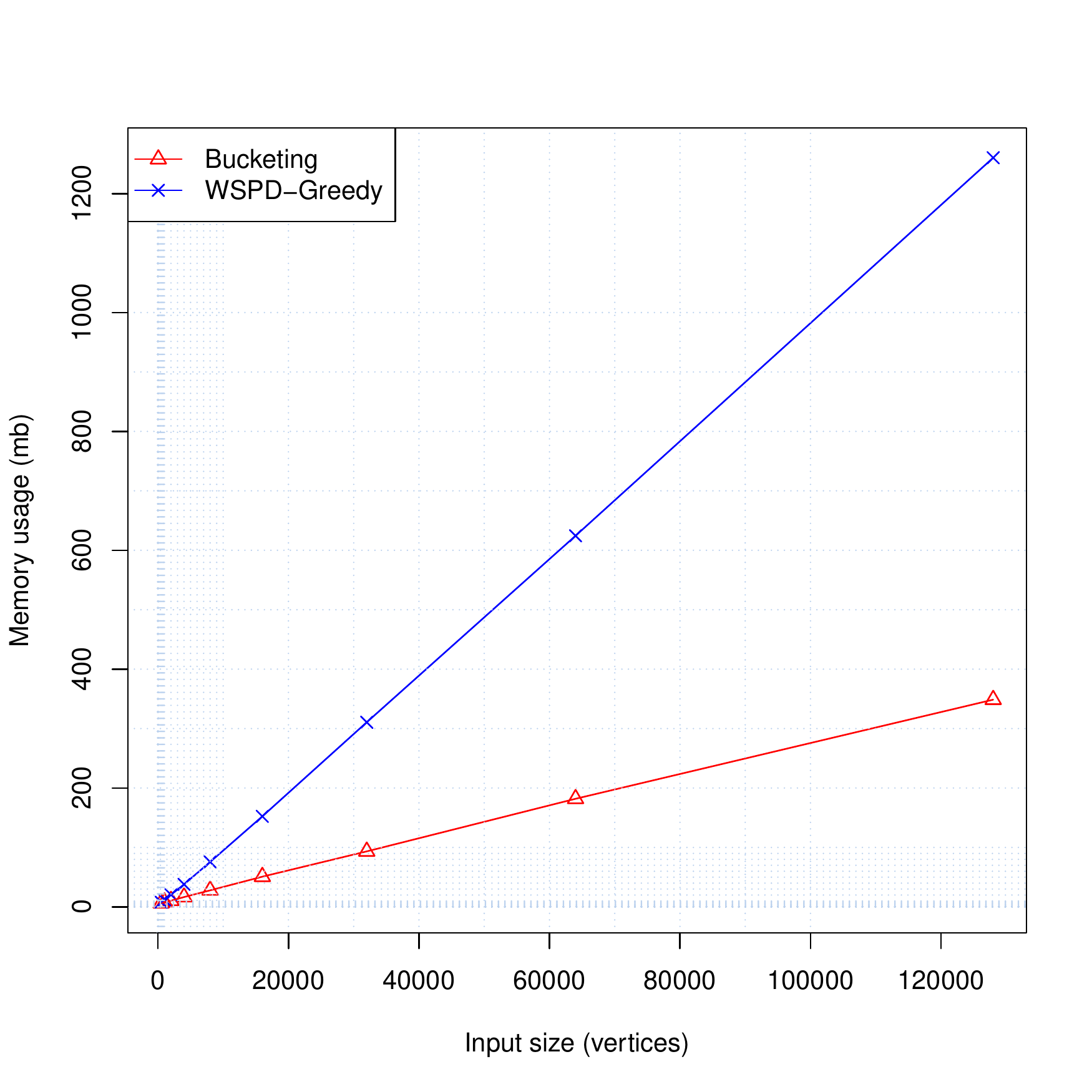}
 \vspace{-.5em}
 \caption{The left plot shows the running time of our algorithm (Bucketing) and WSPD-Greedy for $t=2$ on variously sized point sets generated with a normal distribution. The right plot shows the memory usage on the same data}
 \label{figure:normalplot}
 \end{figure}

 \begin{figure}[h!]\centering
 \includegraphics[width=8cm]{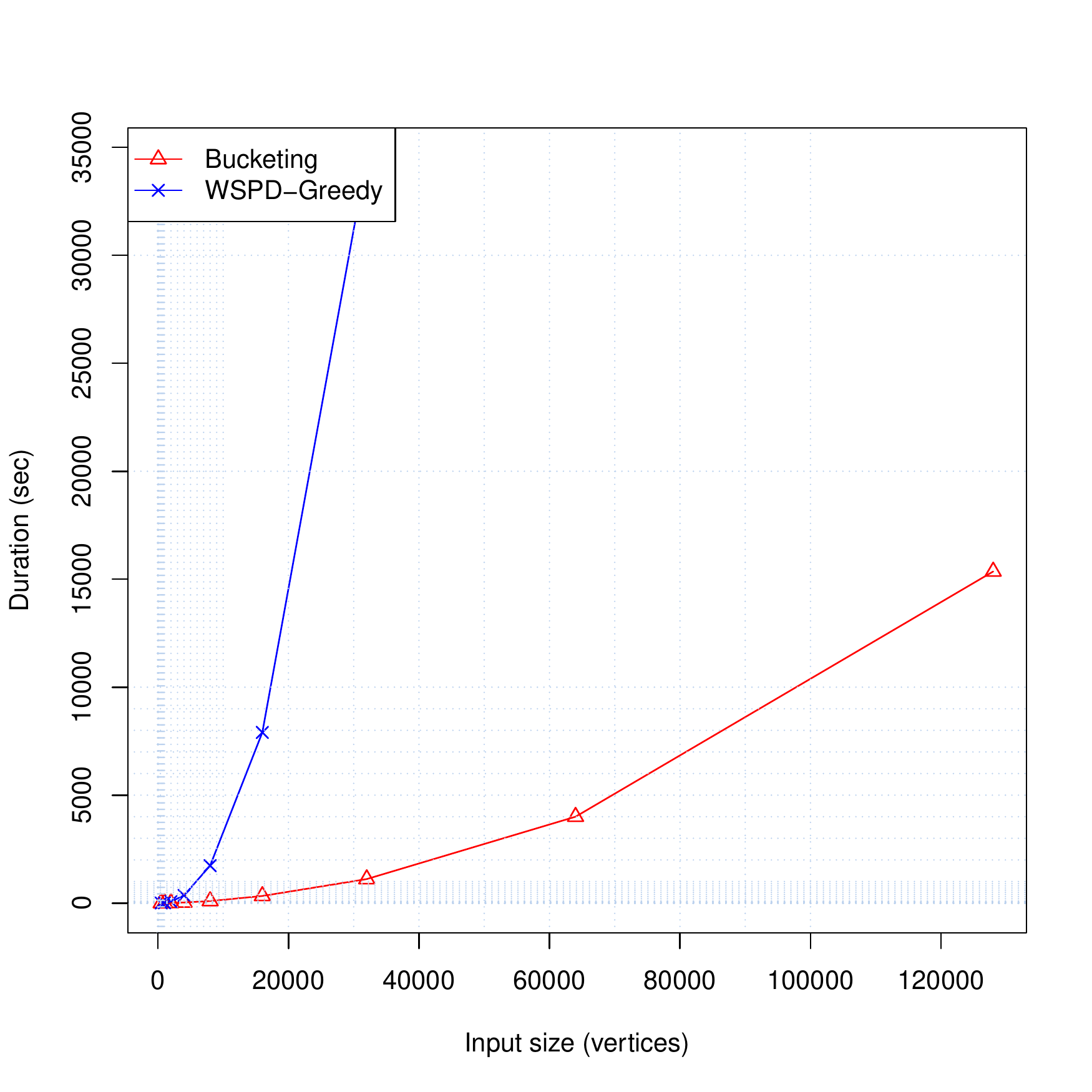}
 \includegraphics[width=8cm]{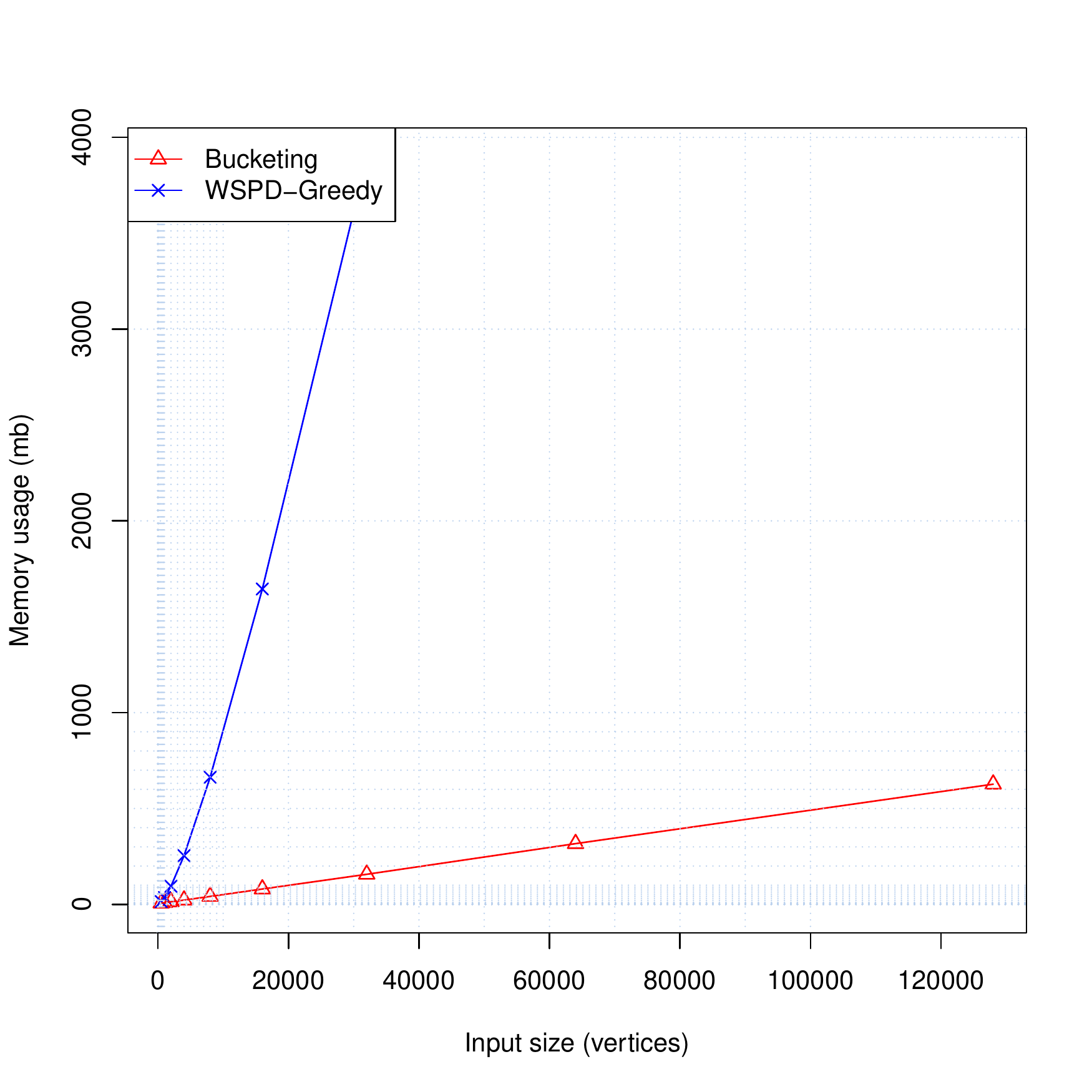}
 \vspace{-.5em}
 \caption{The left plot shows the running time of our algorithm (Bucketing) and WSPD-Greedy for $t=1.1$ variously sized uniformly distributed instances. The right plot shows the memory usage on the same data}
 \label{figure:uniformtplot}
 \end{figure}

 \begin{figure}[h!]\centering
 \includegraphics[width=8cm]{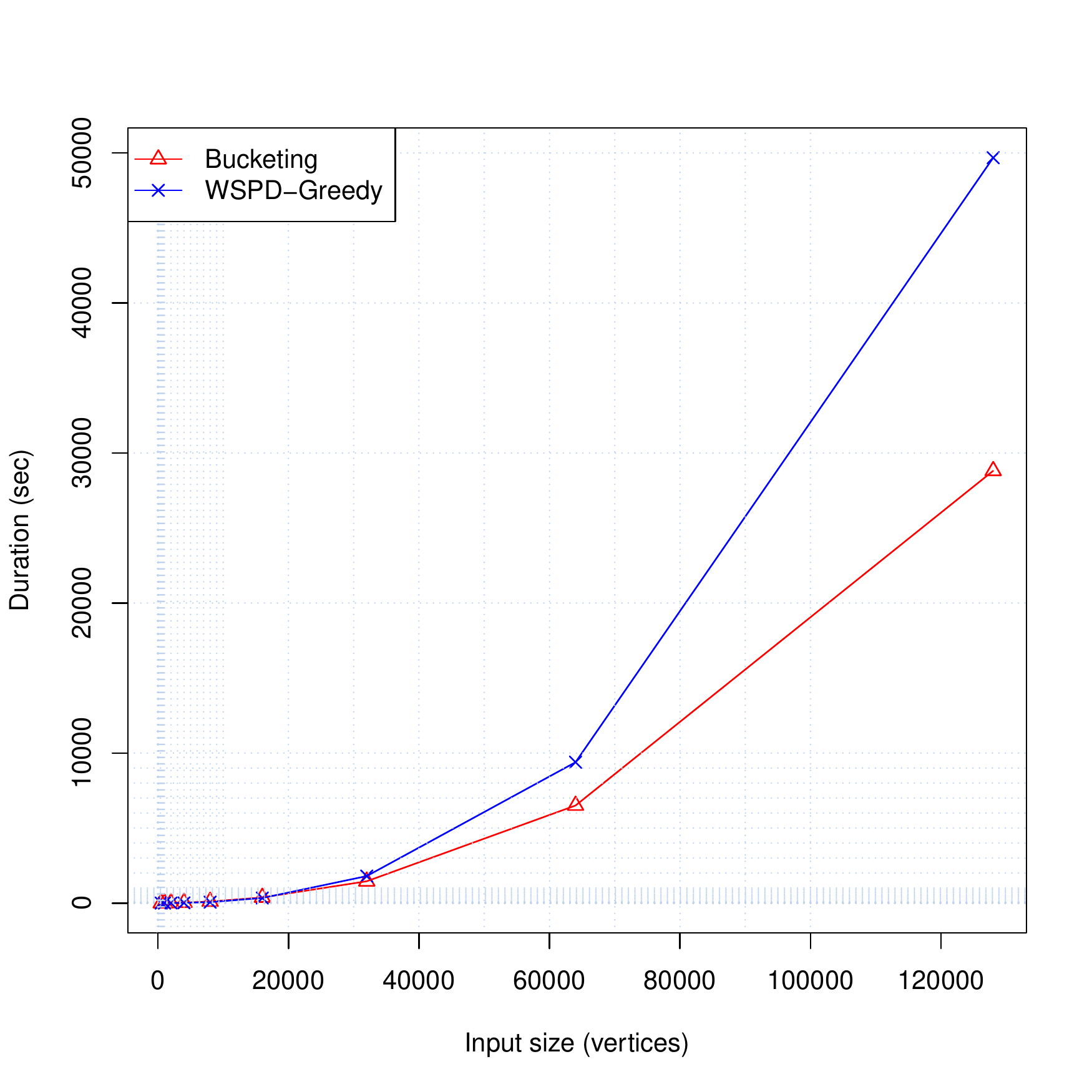}
 \includegraphics[width=8cm]{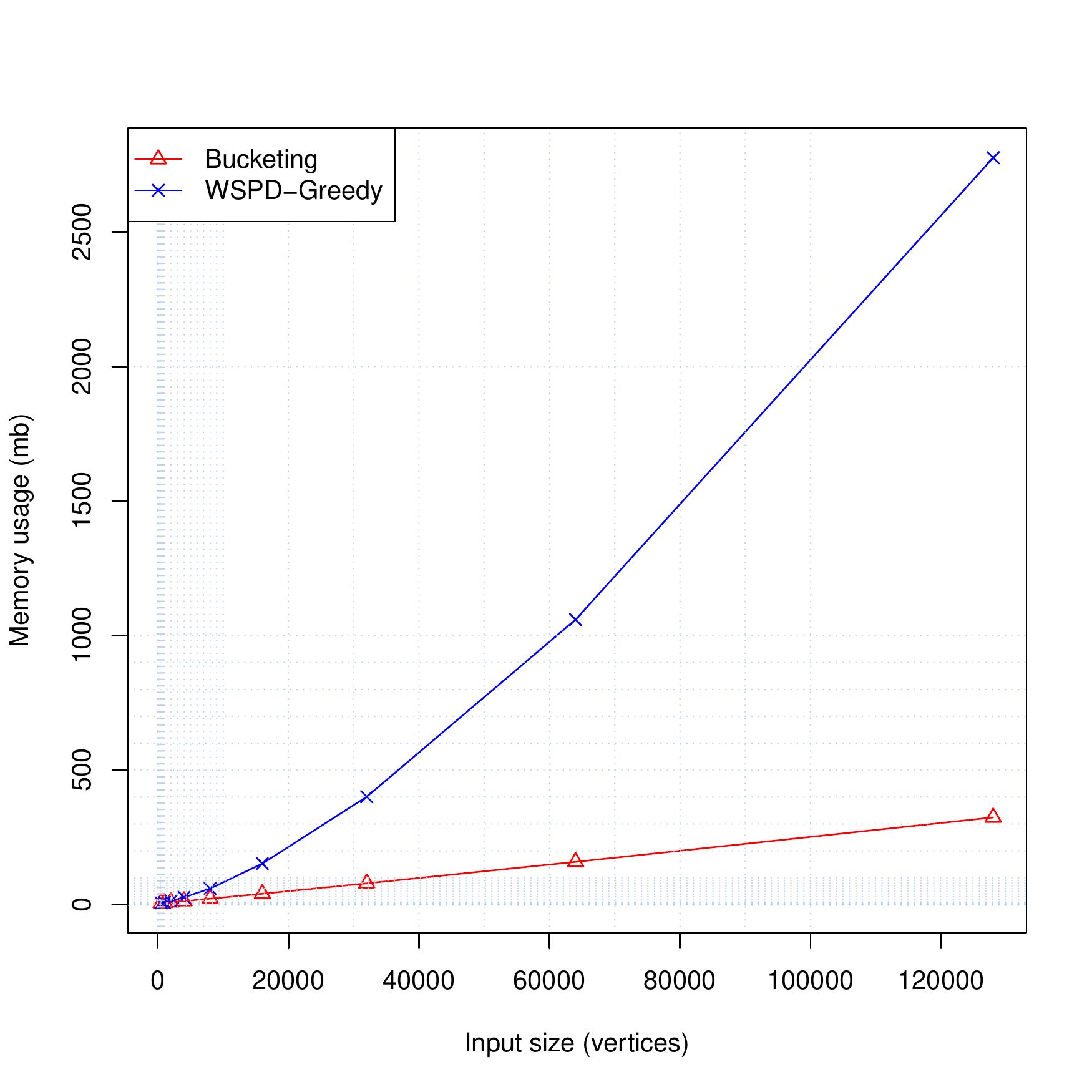}
 \vspace{-.5em}
 \caption{The left plot shows the running time of our algorithm (Bucketing) and WSPD-Greedy for $t=1.1$ variously sized clustered instances. The right plot shows the memory usage on the same data}
 \label{figure:squarestplot}
 \end{figure}

\end{document}